\documentclass[journal]{new-aiaa}

\usepackage[utf8]{inputenc}
\usepackage{newtxtext,newtxmath}

\usepackage{multirow}
\usepackage{xcolor}
\usepackage{indentfirst}
\usepackage{graphicx}

\usepackage{amsmath,amsfonts,amssymb,amsthm,amsbsy}
\usepackage{enumitem}
\usepackage[ruled,linesnumbered]{algorithm2e}
\usepackage{subcaption}

\theoremstyle{plain}
\newcounter{THM}
\newtheorem{thm}[THM]{Theorem}

\theoremstyle{definition}

\theoremstyle{remark}
\newcounter{REM}
\newtheorem{rem}[REM]{Remark}
\newtheorem{remapp}{Remark}[subsection]

\DeclareMathOperator{\diag}{diag}

\allowdisplaybreaks

\title{Synchronisation-Oriented Design Approach for Adaptive Control}

\author{Namhoon Cho\footnote{Research Fellow, Centre for Autonomous and Cyber-Physical Systems, School of Aerospace, Transport and Manufacturing, \texttt{n.cho@cranfield.ac.uk}}}
\affil{Cranfield University, Cranfield, MK43 0AL, Bedfordshire, United Kingdom}
\author{Seokwon Lee\footnote{Assistant Professor, School of Mechanical Engineering, \texttt{seokwonlee@cau.ac.kr}}}
\affil{Chung-Ang University, Seoul, 06974, Republic of Korea}
\author{Hyo-Sang Shin\footnote{Professor, Cho Chun Shik Graduate School of Mobility, \texttt{eshy@kaist.ac.kr}}\footnote{Professor of Guidance, Navigation and Control, Centre for Autonomous and Cyber-Physical Systems, School of Aerospace, Transport and Manufacturing, \texttt{h.shin@cranfield.ac.uk}}}
\affil{Korea Advanced Institute of Science and Technology, Daejeon, 34051, Republic of Korea}
\affil{Cranfield University, Cranfield, MK43 0AL, Bedfordshire, United Kingdom}

\begin{document}
	\maketitle
	
	\begin{abstract}
	This study presents a synchronisation-oriented perspective towards adaptive control which views model-referenced adaptation as synchronisation between actual and virtual dynamic systems. In the context of adaptation, model reference adaptive control methods make the state response of the actual plant follow a reference model. 
	In the context of synchronisation, consensus methods involving diffusive coupling induce a collective behaviour across multiple agents. 
	We draw from the understanding about the two time-scale nature of synchronisation motivated by the study of blended dynamics. The synchronisation-oriented approach consists in the design of a coupling input to achieve desired closed-loop error dynamics followed by the input allocation process to shape the collective behaviour. We suggest that synchronisation can be a reasonable design principle allowing a more holistic and systematic approach to the design of adaptive control systems for improved transient characteristics. Most notably, the proposed approach enables not only constructive derivation but also substantial generalisation of the previously developed closed-loop reference model adaptive control method. 
	Practical significance of the proposed generalisation lies at the capability to improve the transient response characteristics and mitigate the unwanted peaking phenomenon at the same time.
\end{abstract}

\section{Introduction}
Transient performance as well as robustness often govern the choice and tuning of a control system in practice while the asymptotic stability properties are ensured as a prerequisite. The interaction between the uncertainty approximator and the tracking controller along with the physical nature of the uncertainty further complicates the design. Quantifiable characterisation ensuring uniform performance bounds is also challenging.

Development of design methods for improved transient dynamics has been one of the important areas in the research on adaptive control. In particular, the undesirable oscillatory transient response often occurs in Model Reference Adaptive Control (MRAC) systems with the basic direct adaptation law when increasing the rate of adaptation for faster tracking and shorter transient time. The tradeoff between conflicting performance goals leads to the difficulty in tuning. The oscillatory transients can arise from multiple causes; such as 
\begin{enumerate}[label=\roman*)]
	\item lack of damping as in the case of pure integral action leading to overshoot tendency
	\item multivariable nature of the system that can lead to rotational trajectories within the contracting sublevel set of a scalar Lyapunov function
	\item choice of initial condition and external driving input
	\item insufficient time-scale separation between the transient dynamics and the desired reference model
\end{enumerate}
Various approaches have been developed to mitigate the unwanted behaviours depending on the cause of oscillation being addressed; e.g., composite MRAC \cite{Lavretsky2009}, $\mathcal{L}_{1}$ adaptive control \cite{Hovakimyan2010}, dynamic regressor extension and mixing \cite{Ortega2020}, and closed-loop reference model MRAC (CRM-MRAC) \cite{Lavretsky2013}.

With this background, the theory of closed-loop/observer-like reference model has been developed for adaptive control architectures \cite{Lavretsky2013, Gibson2013, Gibson2014, Gibson2015, Qu2020} and showed performance benefits in applications \cite{Wiese2015,Qu2016}. Central to the closed-loop reference model concept is the addition of a tracking error feedback term that resembles the innovation term in a Luenberger observer to the reference model. As a result, the reference model evolves with feedback interaction unlike the fixed open-loop reference model. The key insight is to give distinction between the time constants of the tracking error dynamics and the desired reference model dynamics so that the transients of the former decay more quickly than the latter. Similar understanding was also considered important for the loop shaping via design of the state predictor in $\mathcal{L}_{1}$ adaptive control architecture (see Remark 2.1.1 and Sec. 2.1.6 of \cite{Hovakimyan2010}).

However, the CRM-MRAC method requires careful choice of its design parameters to obtain the expected improvements in transient performance. A bad choice of the combination of learning rate and observer gain may result in significantly degraded performance, e.g., in terms of the input rate. The method with a fixed observer gain is known to suffer from water-bed effects in that using a small observer gain leads to insufficient time-scale separation that may cause high-frequency oscillation while using a large observer gain leads to slow dynamics. This tradeoff between improved transient performance and improved converge rate of parameters poses a major challenge in the design and implementation of the CRM-MRAC systems. In \cite{Yuksek2021}, a policy for variable observer gain was trained by using reinforcement learning to increase the observer gain in the initial phase of the adaptation process for improved transient performance and to decrease the observer gain later for faster system response. Nonetheless, the tuning difficulty is not completely resolved, raising the necessity to revisit the problem.

This study aims to develop a general design approach which can overcome the shortcomings of the existing CRM-MRAC method. The key insight is that the MRAC architecture with a reference model interacting with the plant via feedback connection is conceptually similar to the consensus methods for multi-agent synchronisation. Adaptation based on closed-loop reference model can be formalised as synchronisation between the physical plant and the virtual dynamic system. Motivated by this connection, the main objective is to develop a constructive design approach for adaptive control based on synchronisation with a virtual dynamics playing the role of a closed-loop reference model. 

Synchronisation of multiple agents has been studied to realise distributed autonomous systems. Consensus strategies lay the foundation for distributed computational tasks such as formation control, optimisation, and information fusion \cite{Ren2005a, Ren2005b, Ren2007, Oh2015}. Specifically, recent studies on practical synchronisation of heterogeneous multiple agent system showed that the collective behaviour approximately follows \emph{the average of individual dynamics}, i.e., the blended dynamics, under sufficiently strong diffusive coupling if the blended dynamics is stable \cite{Kim2016, Lee2020, Lee2022b}. This theoretical finding has led to the establishment of a concrete design methodology that first designs the blended dynamics and then assigns each agent to the individual dynamics \cite{Lee2022a, Lee2022c}.


This study presents a synchronisation-oriented design approach for adaptive control by taking insights from the blended dynamics approach. The blended dynamics approach highlights the importance of time-scale separation in achieving synchronisation and generation of desired collective behaviours. Similarly, this study emphasises the two time-scale nature of adaptive control problem which involves synchronisation of two agents. The proposed approach is to design the coupling input first and allocate it to the input variables of plant and virtual dynamics. To clearly explain the main concept, this study considers a linear system with matched uncertainty.

The synchronisation-oriented design approach systematically generalises the CRM-MRAC method developed in \cite{Lavretsky2013, Gibson2013, Gibson2014}. The proposed approach provides additional degrees-of-freedom for adjusting performance through the allocation of coupling input between the plant and the virtual dynamics without changing the tracking error dynamics. It turns out that the existing CRM-MRAC method corresponds to the special case where the diffusive coupling input is allocated only to the virtual dynamics, that is, the tracking error feedback term exists only in the reference model. Unlike the CRM-MRAC architecture of \cite{Lavretsky2013, Gibson2013, Gibson2014}, the proposed approach allows the plant to contribute to form the coupling input. Also, the proposed approach enables prescription of a stable higher-order nominal dynamics for the tracking error through the design of coupling input. In this manner, the notion of synchronisation supersedes the notion of adding tracking error feedback term to the reference model as the central principle for improving transient characteristics of an adaptive control system.

Regarding the practical advantage, the synchronisation-oriented design approach mitigates the tuning difficulty of the CRM-MRAC system due to the tradeoff between transients and system speed. The proposed approach is capable of adjusting the uncertainty cancellation behaviour blending instantaneous uncertainty rejection and online model learning depending on the allocation of coupling input. The weighting factor introduced to define the weighted average of the states representing the collective behaviour determines the allocation strategy. This study shows that choosing this weighting factor to minimise the gap between the virtual dynamics (i.e., the closed-loop reference model) and its ideal counterpart enables alleviation of the unwanted peaking phenomenon observed in the existing CRM-MRAC method with a strong coupling gain (i.e., a large observer gain) due to the decreased convergence rate of parameter adaptation. Such allocation choice can be realised since the plant also actively participates in the synchronisation process. In this way, the proposed approach can improve the transient performance while avoiding the undesirable effects of slow system response because of the additional flexibility, overcoming the limitation of the CRM-MRAC.

The rest of the paper is organised as follows: Section \ref{Sec:SyncAC} presents the synchronisation-oriented approach to the design of MRAC system. Section \ref{Sec:Relations} discusses the proposed approach as a generalisation of the existing CRM-MRAC method. Since this study mainly focuses on highlighting the importance of the structure that allows for more flexibility in coupling input allocation, an overview of the uncertainty approximation methods for adaptive control which is needed for completeness of the paper is deferred to Appendix. Section \ref{Sec:NumSim} provides a numerical simulation showing the effects of different coupling input design and allocation on performance in a direct MRAC setting to demonstrate the additional design flexibility provided by the proposed approach. Section \ref{Sec:Concls} concludes the paper with a summary of remarks.

\section{Synchronisation-Oriented Approach to Adaptive Control} \label{Sec:SyncAC}
The proposed design approach proceeds with the insights gained from distributed consensus algorithms; i) \emph{both} the plant and the reference model in MRAC system can actively participate in the synchronisation process via their inputs, and ii) the input of each system can be defined in various ways without changing the state error dynamics. This section is devoted to the detailed development of the idea into a systematic and generalised approach towards the design of MRAC systems.

\subsection{System Description}
Consider the class of single-input single-output (SISO) systems given by
\begin{equation} \label{Eq:sys_p}
	\Sigma_{p}: \qquad
	\begin{aligned}
	\dot{x}\left(t\right) &= Ax\left(t\right)+b\left(u\left(t\right)+\Delta\left(t\right)\right) +b_{r}r\left(t\right),\quad x\left(0\right)=x_{0}\\
	y\left(t\right) &= c^{T}x\left(t\right)
	\end{aligned}
\end{equation}
where $x\left(t\right)\in\mathbb{R}^{n}$, $u\left(t\right)\in\mathbb{R}$, $y\left(t\right) \in \mathbb{R}$, and $r\left(t\right)\in\mathbb{R}$ represent the state, the control input, the output of the system, and the exogenous reference, respectively; $A \in \mathbb{R}^{n\times n}$, $b\in\mathbb{R}^{n}$, $b_{r}\in\mathbb{R}^{n}$, and $c\in\mathbb{R}^{n}$ are known constant matrices that render $\left(A,b\right)$ controllable; $\Delta\left(t\right) \in \mathbb{R}$ is the bounded matched uncertainty which is the unknown component of $\Sigma_{p}$. We will first proceed without specifying any parameterisation structure for the uncertainty whose impact to the system is confined within the column space of control effectiveness matrix. 

This study considers full state feedback control of the SISO systems affected by matched uncertainties to illustrate the main point of the synchronisation-oriented approach. Note that the theory can be developed for the multi-input multi-output case in a similar manner, but it is not pursued in the present paper. The control design objective is output tracking which requires the output $y\left(t\right)$ to follow a given piecewise-continuous bounded reference signal $r\left(t\right)$ with performance guarantees for both steady-state and transient responses. Another design objective is to perform adaptation without undesirable oscillatory transients.

The plant controller $u\left(t\right)$ can be structured as a composition of the baseline component $u_{base}\left(t\right)$ which achieves stable output tracking in the absence of uncertainty, the adaptive component $u_{ad}\left(t\right)$ that augments the controller to cancel uncertainties, and an additional component $u_{c}\left(t\right)$ that will be designed to achieve synchronisation with the virtual system. The controller is given by
\begin{equation} \label{Eq:u_p}
	u\left(t\right) = u_{base}\left(t\right) + u_{ad}\left(t\right) + u_{c}\left(t\right)
\end{equation}
with
\begin{align}
	u_{base}\left(t\right) &= -k_{m}^{T}x\left(t\right) -k_{r}r\left(t\right) \label{Eq:u_base}\\
	u_{ad}\left(t\right) &= -\hat{\Delta}\left(t\right) \label{Eq:u_ad}
\end{align}
where $k_{m} \in \mathbb{R}^{n}$ satisfies $A_{m} = A - bk_{m}^{T}$ for a Hurwitz constant matrix $A_{m} \in \mathbb{R}^{n \times n}$, $k_{r} \in \mathbb{R}$ satisfying $b_{m}=-bk_{r}+b_{r}$ enforces the unity DC gain of $c^{T}\left(sI-A_{m}\right)^{-1}b_{m}$, and $\hat{\Delta}\left(t\right)$ represents the approximation/estimation of the uncertainty. 

The closed-loop system can be described by substituting Eqs. \eqref{Eq:u_p}-\eqref{Eq:u_ad} into Eq. \eqref{Eq:sys_p} as
\begin{equation} \label{Eq:sys_p_CL}
	\Sigma_{p}^{CL}: ~~
	\begin{aligned}
	\dot{x}\left(t\right) &= A_{m}x\left(t\right)+b_{m}r\left(t\right) + b\left(u_{c}\left(t\right)-\tilde{\Delta}\left(t\right)\right), \quad x\left(0\right) =x_{0}\\
	y\left(t\right) &= c^{T}x\left(t\right)
	\end{aligned}
\end{equation}
where $\tilde{\Delta}\left(t\right):= \hat{\Delta}\left(t\right) - \Delta\left(t\right)$ denotes the uncertainty approximation error.

\begin{rem}
	At this stage, we are considering $u_{c}\left(t\right)$ as a placeholder introduced for further terms in the controller. The distinction between $u_{ad}\left(t\right)$ and $u_{c}\left(t\right)$ might appear notional in that it is possible to define each term to subsume the other term. Nonetheless, this study introduces $u_{c}\left(t\right)$ as a separate term to keep everything associated with uncertainty cancellation within $u_{ad}\left(t\right)$ so that the theoretical development can be more agnostic to the learner, i.e., uncertainty approximator.
\end{rem}

\subsection{Adaptive Control as Two-Agent Synchronisation} \label{SubSec:AC_as_Sync}
This study presents a synchronisation-oriented approach which unfolds its design process based on the perspective that adaptive control can be understood as multi-agent synchronisation. The proposed approach is motivated by the theory of blended dynamics developed in \cite{Kim2016, Lee2020, Lee2022a, Lee2022b, Lee2022c} for distributed synchronisation of multiple heterogeneous agents. For the purpose of adaptive control, this approach introduces a known virtual system $\Sigma_{m}$ as the target for synchronisation with the given physical system $\Sigma_{p}$ to compensate for the effects of uncertainty in $\Sigma_{p}$. Let us consider the virtual dynamic system given by
\begin{equation} \label{Eq:sys_m}
	\Sigma_{m}: \qquad
	\begin{aligned}
	\dot{x}_{m}\left(t\right) &= A_{m}x_{m}\left(t\right)+b_{m}r\left(t\right)+ U_{m}\left(t\right), \quad x_{m}\left(0\right)=x_{m_{0}} \\
	y_{m}\left(t\right) &= c^{T}x_{m}\left(t\right)
	\end{aligned}
\end{equation}
where $U_{m}\left(t\right) \in \mathbb{R}^{n}$ is the control input of the virtual system, and $\left(A_{m},b_{m}\right)$ describes the desired nominal reference dynamics. The virtual system is called the explicit reference model in MRAC literature or the state predictor in $\mathcal{L}_{1}$ adaptive control literature. In the viewpoint of multi-agent control, the virtual system is simply an agent participating in the network.

The general design principle of the proposed approach consists in i) quick synchronisation between the states of $\Sigma_{p}^{CL}$ and $\Sigma_{m}$, and ii) shaping of the collective behaviour to achieve desirable response by \emph{assembling} the agents' closed-loop dynamics. For these purposes, the synchronisation-oriented approach introduces the input signals $u_{c}\left(t\right)$ and $U_{m}\left(t\right)$ into the systems $\Sigma_{p}^{CL}$ in Eq. \eqref{Eq:sys_p_CL} and $\Sigma_{m}$ in Eq. \eqref{Eq:sys_m}, respectively. 

The proposed approach boils down to the design of \emph{both} $u_{c}\left(t\right)$ and $U_{m}\left(t\right)$ to achieve the stated control objectives. The studies on distributed synchronisation for a large number of connected agents usually assume that each agent employs a diffusive coupling input of an identical structure. In contrast, the coupling inputs introduced in $\Sigma_{p}^{CL}$ and $\Sigma_{m}$ are not necessarily of the same form in the adaptive control context. Instead, either of the two dynamic systems can provide parts of the coupling terms that together constitute a desired state error dynamics when assembled. 

It will be shown in Sec. \ref{Sec:Relations} that the proposed approach provides a generalisation of the previously developed MRAC methods with open-loop or closed-loop reference model.

\subsection{Two Time-Scale Nature of Adaptive Control} \label{SubSec:tscale_sep}
Let us define the state error $e\left(t\right)$ and the weighted average of states $z\left(t\right)$ as follows:
\begin{align}
	e\left(t\right) &:= x_{m}\left(t\right) - x\left(t\right) \label{Eq:e_defn}\\
	z\left(t\right) &:= \mu x\left(t\right) + \left(1-\mu\right) x_{m}\left(t\right) = -\mu e\left(t\right) + x_{m}\left(t\right) \label{Eq:z_defn}
\end{align}
where $0 \leq \mu \leq 1$ is a constant weighting factor. The state $x\left(t\right)$ and $x_{m}\left(t\right)$ of the plant and the virtual dynamic system, respectively, can be represented in terms of the state error $e\left(t\right)$ and the weighted average of states $z\left(t\right)$ as
\begin{align} 
	x\left(t\right) & = z\left(t\right) - (1-\mu)e\left(t\right) \label{Eq:x_coord_trans}\\
	x_{m}\left(t\right) &= z\left(t\right)+\mu e\left(t\right) \label{Eq:x_m_coord_trans} 
\end{align}
showing that the variables are related to each other through an invertible coordinate transformation. 

\begin{rem}
The above change of coordinate relation suggests that we can analyse the properties of plant state response via the knowledge about $e\left(t\right)$ and $z\left(t\right)$. From Eq. \eqref{Eq:x_coord_trans}, we have
\begin{equation} \label{Eq:tri_ieq}
	\begin{aligned}
	\left\|x\left(t\right)\right\| &\leq \left\|z\left(t\right)\right\| + \left(1-\mu\right)\left\|e\left(t\right)\right\|\\
	&\leq \sup\limits_{\tau \in \left[0,t\right]}\left\|z\left(\tau\right)\right\| + \left(1-\mu\right)\sup\limits_{\tau \in \left[0,t\right]}\left\|e\left(\tau\right)\right\|
	\end{aligned}
\end{equation}
for any proper choice of vector norm verifying the triangle inequality. The relation shows that the uniform bound of $x\left(t\right)$ can be estimated from the bounds of transformed signals $e\left(t\right)$ and $z\left(t\right)$.
\end{rem}

The dynamics of state error can be obtained by taking the difference between Eqs. \eqref{Eq:sys_p_CL} and \eqref{Eq:sys_m} as
\begin{equation} \label{Eq:sys_e}
	\Sigma_{e}=\Sigma_{m}-\Sigma_{p}^{CL}: \qquad
	\begin{aligned}
	\dot{e}\left(t\right) &= A_{m}e\left(t\right) + U_{m}\left(t\right) -bu_{c}\left(t\right) + b\tilde{\Delta}\left(t\right)\\
	e\left(0\right) &=  x_{m_{0}} - x_{0}
	\end{aligned}
\end{equation}
Likewise, the dynamics of weighted state average can be obtained by the linear combination of Eqs. \eqref{Eq:sys_p} and \eqref{Eq:sys_p_CL} as
\begin{equation} \label{Eq:sys_z}
	\Sigma_{z}=\mu\Sigma_{p}^{CL}+\left(1-\mu\right)\Sigma_{m}:\qquad
	\begin{aligned}
	\dot{z}\left(t\right) &= A_{m}z\left(t\right) + b_{m}r\left(t\right) +\mu bu_{c}\left(t\right) + \left(1-\mu\right)U_{m}\left(t\right) - \mu b\tilde{\Delta}\left(t\right)\\
	z\left(0\right) &= \mu x_{0} + \left(1-\mu\right)x_{m_{0}}
	\end{aligned}
\end{equation}
It should be noted that the system $\Sigma_{z}$ has discrepancies from the following ideal reference-tracking dynamics.
\begin{equation} \label{Eq:sys_id}
	\Sigma_{id}: \qquad 
	\begin{aligned}
	\dot{x}_{id}\left(t\right) &= A_{m}x_{id}\left(t\right)+b_{m}r\left(t\right)\\
	x_{id}\left(0\right)&= \mu x_{0} + \left(1-\mu\right)x_{m_{0}}
	\end{aligned}
\end{equation}

\begin{rem}
The influence of uncertainty cancellation error $\tilde{\Delta}\left(t\right)$ appears to be reduced in $\Sigma_{z}$ as compared to its effect in the closed-loop plant dynamics $\Sigma_{p}^{CL}$. This is coherent with the observation that the collective behaviour obtained with the blended dynamics approach has certain robustness against external disturbance and parametric variations \cite{Lee2022a}.
\end{rem}


If the synchronisation can be achieved relatively earlier in comparison to the convergence of individual states $x\left(t\right)$ and $x_{m}\left(t\right)$, then we can approximate that $z\left(t\right) \approx x\left(t\right) \approx x_{m}\left(t\right)$, which implies that $x\left(t\right)$ will closely follow the solution $z\left(t\right)$ of the blended dynamics $\Sigma_{z}$. In this sense, one can expect that the transient and steady-state response of the plant state will closely follow the ideal command-following dynamics when both synchronisation and collective behaviour shaping goals (see (C1) and (C2) in Sec. \ref{SubSubSec:DesignGuideline}) are achieved. 

Therefore, a reasonable design approach is to make $\Sigma_{e}$ be the fast time-scale dynamics and let $\Sigma_{z}$ be the slow time-scale dynamics. From this point of view, the term $\mu bu_{c}\left(t\right) +\left(1-\mu\right)U_{m}\left(t\right)$ in Eq. \eqref{Eq:sys_z} can be viewed as a fast-varying but small, bounded, and vanishing perturbation signal acting on the system $\Sigma_{z}$ since the allocated signals $u_{c}\left(t\right)$ and $U_{m}\left(t\right)$ will be a function of the fast decaying signal $e\left(t\right)$.

In this way, the synchronisation-oriented viewpoint naturally accounts for the two time-scale nature of practical adaptive control with improved transients.

\subsection{Coupling Input Allocation}
\subsubsection{Design Considerations} \label{SubSubSec:DesignGuideline}
Motivated by the observations made in Sec. \ref{SubSec:AC_as_Sync} and \ref{SubSec:tscale_sep}, one can formulate a constrained optimisation problem for synthesis of $u_{c}\left(t\right)$ and $U_{m}\left(t\right)$ considering the following requirements:
\begin{enumerate}[label = (C\arabic*)]
	\item For synchronisation, $u_{c}\left(t\right)$ and $U_{m}\left(t\right)$ should render $\Sigma_{e}$ 
	an exponentially stable system in the absence of uncertainty, or at least a practically convergent system considering the presence of uncertainty. 
	
	\item For collective behaviour shaping, $u_{c}\left(t\right)$ and $U_{m}\left(t\right)$ should render $\Sigma_{z}$ 
	as close as possible to $\Sigma_{id}$.

	\item The component allocated to $\Sigma_{p}$ should lie in the column space of plant control effectiveness matrix $b$ since the plant input $u\left(t\right)$ can affect state change only along the direction of $b$.

	\item The finite bandwidth of the physical actuator places an upper bound on the frequency up to which commanded inputs can be achieved. Unlike the virtual dynamic system $\Sigma_{m}$ which can perfectly realise a commanded control input, the control action for the actual plant $\Sigma_{p}$ is always subject to the hardware limit.

	\item More implicit practical considerations include the necessity to avoid oscillatory time-domain responses, the requirement to keep signals in the feedback loop (uniformly) bounded, and the importance of loop shaping for assuring robustness.
\end{enumerate}

\subsubsection{Desired Synchronisation Dynamics} \label{SubSubSec:Sigma_e_CL}
The state error dynamics given in Eq. \eqref{Eq:sys_e} is already input-to-state stable with respect to bounded $\tilde{\Delta}\left(t\right)$ as the driving input for zero $u_{c}\left(t\right)$ and $U_{m}\left(t\right)$ since $A_{m}$ is Hurwitz. However, the eigenvalues of $A_{m}$ usually tend to be located at moderately slow region as the baseline controller is designed for the nominal plant to avoid abrupt time-domain response and also to meet certain robustness requirements in terms of stability margins. 

Thus, the term $U_{m}\left(t\right)-bu_{c}\left(t\right)$ can be prescribed by specifying the desired state error dynamics i) to improve the overall response rate of $\Sigma_{e}$ without inducing overshooting behaviour that may manifest itself as oscillatory transients when propagated through the loop, ii) to limit the impact of nonzero $\tilde{\Delta}\left(t\right)$ on $e\left(t\right)$, and iii) to make the controlled response scale with the reference command for overall consistency and predictability. To meet these loop shaping goals, a reasonable design is to set the desired state error dynamics by a linear time-invariant (LTI) system.

Considering the possibility to increase the system order, the integral augmented state error dynamics can be written as
\begin{equation} \label{Eq:sys_e_CL}
	\Sigma_{e}^{CL}: \qquad 
	\begin{aligned}
		\dot{\mathbf{e}}_{I}^{l}\left(t\right) &= A_{e}\mathbf{e}_{I}^{l}\left(t\right) + B_{e}\left(U_{c}\left(t\right)+b\tilde{\Delta}\left(t\right)\right)\\
		\mathbf{e}_{I}^{l}\left(0\right) &= \mathbf{e}_{I_{0}}^{l}
	\end{aligned}
\end{equation}
with
\begin{equation} \label{Eq:Ae_Be}
	A_{e} = \begin{bmatrix}
		A_{m}	& 0		&\cdots & 0 & 0\\
		I 	& 0 	&\cdots & 0 & 0\\
		0	& I	 	&\cdots	& 0 & 0\\
	\vdots  & \vdots &\ddots & \vdots &\vdots\\
		0	& 0		&\cdots & I & 0
	\end{bmatrix},\quad
	B_{e} = \begin{bmatrix}
		I\\
		0\\
		0\\
		\vdots\\
		0
	\end{bmatrix}
\end{equation}
where 
\begin{equation} \label{Eq:e_I_vec}
	\mathbf{e}_{I}^{l}\left(t\right) = \begin{bmatrix}
		e^{\left<0\right>}\left(t\right)\\
		e^{\left<1\right>}\left(t\right)\\
		\vdots\\
		e^{\left<l\right>}\left(t\right)
	\end{bmatrix}
\end{equation}
for some $l\geq 0$ is the vector aggregating the error integrals defined by
\begin{equation} \label{Eq:e_I}
	\begin{aligned}
		e^{\left<0\right>}\left(t\right) &= e\left(t\right)\\
		\dot{e}^{\left<i\right>}\left(t\right) &= e^{\left<i-1\right>}\left(t\right), ~~ e^{\left<i\right>}\left(0\right) = 0, ~\text{for}~ i=1,2,\ldots,l
	\end{aligned}
\end{equation}
and 
\begin{equation} \label{Eq:U_c_defn}
	U_{c}\left(t\right) := U_{m}\left(t\right) -bu_{c}\left(t\right)
\end{equation}
is the coupling input.

Obviously, a state feedback form design suffices the purpose of stable synchronisation.
\begin{equation} \label{Eq:U_c_design}
	U_{c}\left(t\right)=-K_{e}\mathbf{e}_{I}^{l}\left(t\right)
\end{equation}
Depending on the order of desired error dynamics, the simplest choice is the proportional (P) or proportional-integral (PI) coupling given by
\begin{equation} \label{Eq:alloc_constr}
	U_{c}\left(t\right) =\begin{cases}
		-k_{P} e\left(t\right) & \text{(P coupling)}\\
		-k_{P} e\left(t\right) -k_{I}e^{\left<1\right>}\left(t\right) & \text{(PI coupling)}
	\end{cases}
\end{equation}
for positive scalar gains $k_{P}$ and $k_{I}$. 

The coupling gains are the design parameters that can be chosen considering the following practical requirements:
\begin{enumerate}[label = (G\arabic*)]
	\item For the case of P coupling, the state error dynamics becomes a first-order system given by
	\begin{equation} \label{Eq:sys_e_1}
		\Sigma_{e}^{CL1}:~ \dot{e}\left(t\right)+\left(k_{P}I-A_{m}\right)e\left(t\right) = b\tilde{\Delta}\left(t\right)
	\end{equation} 
	The design parameter $k_{P}$ determines the bandwidth $\omega_{e}$ of $\Sigma_{e}^{CL1}$. Considering the fact that the poles of $A_{m}$ are relatively slow in most practical applications, the approximate relation $k_{P} \approx \omega_{e}$ is justified for a high coupling gain $k_{P}\gg\left|\lambda\left(A_{m}\right)\right|$. Increasing $k_{P}$ will result in i) higher overall rate of convergence in the response of $e\left(t\right)$, ii) dominance of $k_{P}$ over the slow poles of $A_{m}$ in governing the response of $e\left(t\right)$, and iii) reduced error bound $\left\|e\left(t\right)\right\|_{\mathcal{L}_{\infty}}$ assuming boundedness of $\tilde{\Delta}\left(t\right)$.

	\item For the case of PI coupling, the state error dynamics becomes a second-order system given by
	\begin{equation} \label{Eq:sys_e_2}
		\Sigma_{e}^{CL2}: ~ \dot{e}\left(t\right)+\left(k_{P}I-A_{m}\right)e\left(t\right)+k_{I}\int e\left(\tau\right)d\tau = b\tilde{\Delta}\left(t\right)
	\end{equation} 
	The two design parameters $k_{P}$ and $k_{I}$ determine the natural frequency $\omega_{n}$ and the damping ratio $\zeta$ of $\Sigma_{e}^{CL2}$. Again, a reasonable approximation considering the time-scale of system matrix $A_{m}$ representing the ideal command-following dynamics is to state that $\omega_{n} \approx \sqrt{k_{I}}$ and $\zeta \approx \frac{k_{P}}{2\sqrt{k_{I}}}$ hold with high coupling gain $k_{P}\gg\left|\lambda\left(A_{m}\right)\right|$. Similar to the case of P coupling, increasing $\omega_{n}$ will result in i) faster response in $e\left(t\right)$, ii) reduced error bound $\left\|e\left(t\right)\right\|_{\mathcal{L}_{\infty}}$ assuming boundedness of $\tilde{\Delta}\left(t\right)$ and $\dot{\tilde{\Delta}}\left(t\right)$. The choice of the damping ratio is critically important to avoid oscillatory signals entering the control loop.
	
	\item 
	The maximum acceptable bandwidth of the closed-loop error dynamics $\Sigma_{e}^{CL}$ is practically bounded from above by the lowest frequency of the signal content in $\tilde{\Delta}\left(t\right)$. This bandwidth limitation is required to effectively filter out the uncertainty approximation error $\tilde{\Delta}\left(t\right)$ entering into the state error dynamics as well as the slower blended dynamics. If the uncertainty is approximated (for instantaneous rejection as it will be explained in Sec. \ref{SubSec:DOB}) by $\hat{\Delta}\left(t\right) = C\left(s\right)\Delta\left(t\right)$ where $C\left(s\right)$ is a strictly proper and stable scalar low-pass filter with bandwidth of $\omega_{f}$, then $\tilde{\Delta}\left(t\right)$ represents the high-pass-filtered signal lying in the frequency region above $\omega_{f}$. In this case, the coupling gains $k_{P}$ and $k_{I}$ can be chosen so that the bandwidth $\omega_{e}$ of $\Sigma_{e}$ does not exceed $\omega_{f}$.
\end{enumerate}

\subsubsection{Best Possible Blended Dynamics} \label{SubSubSec:best_Sigma_z}
The coupling input $U_{c}\left(t\right)$ governing the synchronisation dynamics can be \emph{allocated} to the input of each system in a manner similar to control allocation for overactuated systems. The transient response characteristics depends on the allocation choice even if the synchronisation dynamics remains the same. Therefore, the allocation should be performed so that the controlled plant exhibits desired performance characteristics.

The problem of coupling input allocation is to determine the pair of $u_{c}\left(t\right)$ and $U_{m}\left(t\right)$ subject to the hard constraint given by Eq. \eqref{Eq:U_c_design}, which is an underdetermined system of linear equations. The number of decision variables is $n+1$ while the constraint specifies $n$ equations. The number of excess degrees-of-freedom remaining after enforcing the constraint is equal to the dimension of plant control input. This additional degree-of-freedom can be exploited to shape the closed-loop response by minimising a physically relevant convex objective function.

One reasonable direction is to keep the response of $\Sigma_{z}$ as close as possible to $\Sigma_{id}$ through pointwise minimisation of the deviation between $\dot{z}$ and $\dot{x}_{id}$. Given a desired coupling input $U_{c}$, the associated optimisation problem can be written as
\begin{equation}\label{Eq:prob_J_pert_con}
	\begin{aligned}
		\underset{u_{c}, U_{m}}{\text{minimise}}~~~   & J_{pert} = \left\| W\left\{\mu b u_{c}\left(t\right) + \left(1-\mu\right)U_{m}\left(t\right)\right\}\right\|_{p}\\
		\text{subject to}~~  & U_{m}-bu_{c} = U_{c}
	\end{aligned}
\end{equation}
where $W>0$ and $p\in \left[1,\infty\right]$. By substituting the constraint relation $U_{m} = U_{c} + bu_{c}$ into the objective function, the equivalent unconstrained problem is to minimise
\begin{equation}\label{Eq:prob_J_pert}
	J_{pert} = \left\| W\left\{b u_{c}  + \left(1-\mu\right)U_{c}\right\}\right\|_{p}
\end{equation}
at each instance. Modern convex programming tools provide efficient computational means of solution. 

The optimal allocation can be expressed in closed-form for the case of $p=2$ where the optimality condition can be described by
\begin{equation}\label{Eq:J_pert_optcond}
	\frac{\partial J_{pert}}{\partial u_{c}} = \frac{\left(W\left\{b u_{c}  + \left(1-\mu\right)U_{c}\right\}\right)^{T}}{\left\|W\left\{b u_{c}  + \left(1-\mu\right)U_{c}\right\}\right\|_{2}}Wb=0
\end{equation}
Solving Eq. \eqref{Eq:J_pert_optcond} for $u_{c}$ gives the optimal solution as
\begin{equation} \label{Eq:sol_J_pert}
	\begin{aligned}
		u_{c}^{*} &= -\left(1-\mu\right)\left(b^{T}W^{T}Wb\right)^{-1}b^{T}W^{T}WU_{c}\\
		U_{m}^{*} &= U_{c} + bu_{c}^{*} 
	\end{aligned}	
\end{equation}

\begin{rem}
	A desired blended dynamics for $z$ can be prescribed as the hard constraint for coupling input allocation instead of enforcing a desired synchronisation dynamics for $e$ if collective behaviour shaping takes priority. 
\end{rem}


\subsection{Uncertainty Approximation}
Imperfect uncertainty approximation results in nonzero $\tilde{\Delta}\left(t\right)$ that acts like a perturbation term driving the closed-loop error dynamics, as it is clearly shown in Eqs. \eqref{Eq:sys_e_1} and \eqref{Eq:sys_e_2} for $\Sigma_{e}^{CL1}$ and $\Sigma_{e}^{CL2}$, respectively. The effect of nonzero $\tilde{\Delta}\left(t\right)$ driving the slow time-scale dynamics $\Sigma_{z}$ described in Eq. \eqref{Eq:sys_z} is also prevalent in that the actual plant response may exhibit deviation from $\Sigma_{id}$ even if the synchronisation is achieved with a faster rate. 

There are two possible directions to keep the actual trajectory close to that expected from the ideal unperturbed synchronisation dynamics considering the dynamic systems $\Sigma_{e}^{CL}$ and $\Sigma_{z}$ as the operators acting on the signal $\tilde{\Delta}\left(t\right)$. One is to reduce the magnitude of $\tilde{\Delta}\left(t\right)$ either at each instance or over time. Another is to confine the frequency contents of $\tilde{\Delta}\left(t\right)$ to a region above the bandwidth $\omega_{e}$ of fast time-scale dynamics $\Sigma_{e}^{CL}$ so that the slow time-scale dynamics $\Sigma_{z}$ shows stronger attenuation over the frequency region occupied by $\tilde{\Delta}\left(t\right)$ (See also (G3) in Sec. \ref{SubSubSec:Sigma_e_CL}).

The algorithm for uncertainty approximation, i.e., the process of constructing $\hat{\Delta}\left(t\right)$, has not been specified up to this point as this study mainly focuses on careful design of the control system architecture rather than on the uncertainty approximator. To clearly show the strength of the proposed approach, the basic direct adaptation algorithm for updating the estimated parameter $\hat{\theta}\left(t\right)$ will be considered later in Sec. \ref{Sec:NumSim} assuming that the uncertainty is linearly parameterised as $\Delta\left(t\right)=\Phi\left(x\left(t\right)\right)^{T}\theta$ where $\Phi\left(x\left(t\right)\right)$ denotes the known basis function. Readers are referred to Appendices \ref{SubSec:DOB} and \ref{SubSec:ML} for an overview of the two broadly defined approaches for uncertainty approximation in adaptive control; instantaneous rejection of lumped uncertainty, and online learning of uncertainty model. Although the online learning method will not be considered in the rest of this paper, Appendix \ref{SubSubSec:stable_adapt} provides the stability analysis for the proposed adaptive control system with online-learning-based adaptation for generality of the result.

\section{Relation to Existing Model Reference Adaptive Control Methods} \label{Sec:Relations}
This section aims to show that the proposed approach provides a general framework which can generate the previously developed MRAC methods as special instances. Each instance corresponds to a certain combination of the coupling input design, the coupling input allocation, and the uncertainty approximation method. The cases with $u_{c}\left(t\right)\equiv 0$ are considered to discuss the relations.

\subsection{Direct MRAC with Open-Loop Reference Model} \label{SubSec:ORM_MRAC}
The classical direct MRAC algorithm with an open-loop reference model, i.e., fixed trajectory generator, corresponds to the following design choice:
\begin{enumerate}
	\item Coupling input design: N/A
	\[U_{c}\left(t\right)=0\] 
	\item Coupling input allocation: N/A
	\[u_{c}\left(t\right)=0, \quad U_{m}\left(t\right)=0\]
	\item Uncertainty approximation: direct adaptation
	\[\hat{\Delta}\left(t\right) = \Phi\left(x\left(t\right)\right)^{T}\hat{\theta}\left(t\right), \quad \dot{\hat{\theta}}\left(t\right) = -\Gamma \Phi\left(x\left(t\right)\right) e\left(t\right)^{T}Pb\]
\end{enumerate}
This basic form can achieve the goal of output tracking in the steady-state response with appropriate adaptation mechanism. However, undesirable high-frequency oscillation is often observed in its transient response, especially when a high adaptation gain is used. 

\subsection{Direct MRAC with Closed-Loop Reference Model} \label{SubSec:CRM_MRAC}
The idea of the observer-like reference model in \cite{Lavretsky2013} also known as the closed-loop reference model in \cite{Gibson2013, Gibson2014, Gibson2015, Qu2020, Wiese2015, Qu2016} has been suggested as a design modification that enables smoother transients by making the error dynamics evolve faster than the desired reference model. The introduction of tracking error feedback in the reference model has been perceived critical to its successes. More specifically, the existing state feedback CRM-MRAC architecture presented in \cite{Gibson2013} corresponds to the following design choice:
\begin{enumerate}
	\item Coupling input design: P coupling with $k_{P}>0$
	\[U_{c}\left(t\right)=-k_{P}e\left(t\right)\] 
	\item Coupling input allocation: $\min J_{pert}$ with $\mu=1$
	\[u_{c}\left(t\right)=0, \quad U_{m}\left(t\right)=U_{c}\left(t\right)\]
	\item Uncertainty approximation: direct adaptation
	\[\hat{\Delta}\left(t\right) = \Phi\left(x\left(t\right)\right)^{T}\hat{\theta}\left(t\right), \quad \dot{\hat{\theta}}\left(t\right) = -\Gamma \Phi\left(x\left(t\right)\right) e\left(t\right)^{T}Pb\]
\end{enumerate}
The CRM-MRAC shows improved transient behaviour as compared to the open-loop reference model MRAC, however, the state response of the virtual system $\Sigma_{m}$ may show significant departure from the desired nominal reference dynamics with a large $U_{m}\left(t\right)$. 

\subsection{New MRAC Methods Based on Proposed Framework}
The synchronisation-oriented approach developed in Sec. \ref{Sec:SyncAC} generalises the existing MRAC methods, suggesting the possibility to explore a new area for further improvements. One can make different choices from those shown above for each of the elements to create design variations. For example, selecting PI coupling instead of P coupling will provide a second-order tracking error dynamics. Unlike the framework in \cite{Gibson2013}, the present study introduces an input component responsible for diffusive coupling not only in the virtual dynamics but also in the plant controller via $u_{c}\left(t\right)$. Different coupling input allocation strategy can affect the transient performance and robustness of the resulting collective behaviour while maintaining the same synchronisation dynamics. Indeed, the additional degrees-of-freedom provided in the proposed approach can be utilised to overcome the shortcomings of the CRM-MRAC scheme by minimising the input contribution of $\Sigma_{m}$ for synchronisation. (See also Sec. \ref{Sec:NumSim} for the numerical simulation comparing different allocations of the coupling input.) Furthermore, parameter update mechanism for direct adaptation can be replaced with a stable online parameter learning algorithm to perform composite adaptation.



\section{Numerical Simulation} \label{Sec:NumSim}
This section aims to show the practical value of the synchronisation-oriented approach to design an adaptive control system. The additional design flexibility introduced in coupling input allocation can be exploited to shape the plant state response without changing the error dynamics. The experiment 
will illustrate that the proposed approach takes the advantage and overcomes the drawback of the CRM-MRAC method.



\subsection{System Model}
The short period longitudinal dynamics of F-16 aircraft trimmed at $\left(V_{T},h,\bar{Q},CG\right)=\left(502\text{ft/s}, 0\text{ft}, 300\text{lb/ft}^{2}, 0.35\bar{c}\right)$ is considered for simulation. The system model which was previously used in \cite{Lavretsky2009} for angle-of-attack tracking can be expressed as
\begin{equation} \label{Eq:sys_sim}
	\underbrace{
		\begin{bmatrix}
			\dot{\alpha}\\
			\dot{q}\\
			\dot{e}_{\alpha_{I}}
		\end{bmatrix}
	}_{\dot{x}}
	=
	\underbrace{
		\begin{bmatrix}
			-1.0189	&	0.9051	&	0\\
			0.8223	&	-1.0774	&	0\\
			1		&	0		&	0
		\end{bmatrix}
	}_{A}
	\underbrace{
		\begin{bmatrix}
			\alpha\\
			q\\
			e_{\alpha_{I}}
		\end{bmatrix}
	}_{x}
	+ 
	\underbrace{
		\begin{bmatrix}
			-0.0022\\
			-0.1756\\
			0
		\end{bmatrix}
	}_{b}
	\left[
		\underbrace{\delta_{e}}_{u} 
		+ 
		\underbrace{
			\begin{bmatrix}
				\alpha \\
				q \\
				e^{-\frac{\left(\alpha - \pi/90\right)^{2}}{2 \cdot 0.0233^{2}}}
			\end{bmatrix}^{T}
		}_{\Phi\left(x\right)^{T}} 
		\underbrace{
			\begin{bmatrix}
				-4.6839\\
				-9.8197\\
				1
			\end{bmatrix}
		}_{\theta}
	\right]
	+
	\underbrace{
		\begin{bmatrix}
			0\\
			0\\
			-1
		\end{bmatrix}
	}_{b_{r}}
	\underbrace{\alpha_{cmd}}_{r}
\end{equation}
\begin{equation} \label{Eq:sys_sim_output}
	\underbrace{\alpha}_{y} 
	= 
	\underbrace{
		\begin{bmatrix}
			1 & 0 & 0
		\end{bmatrix}
	}_{c^{T}}
	\underbrace{
		\begin{bmatrix}
			\alpha\\
			q\\
			e_{\alpha_{I}}
		\end{bmatrix}
	}_{x}
\end{equation}
where the angle-of-attack $\alpha$ is in [rad], the pitch rate $q$ is in [rad/s], and the elevator deflection $\delta_{e}$ is in [deg]. The augmented model in Eq. \eqref{Eq:sys_sim} includes the integral of tracking error $e_{\alpha_{I}}$ as a state variable and the angle-of-attack command $\alpha_{cmd}$ as an exogenous driving input. The basis function $\Phi\left(x\right)$ is assumed to be known.

\subsection{Simulation Setup}
The experiment considers the adaptive control system employing the P coupling for synchronisation and the direct adaptation law for parameter update. Minimisation of $J_{pert}$ is the strategy considered for coupling input allocation since $J_{pert}$ defines a family of objectives depending on $\mu$ as described in Sec. \ref{SubSubSec:best_Sigma_z}. 

Each test case takes a different design and allocation of the coupling input by varying the coupling gain $k_{P}$ and the weighting factor $\mu$ while using fixed values for all other parameters including the adaptation gain $\Gamma$. This setup aims to show how the transient performance depends on synchronisation-related parameters. The effect of different allocation can be studied by comparing the results across different $\mu$ for each $k_{P}$ as the nominal error dynamics solely depends on $k_{P}$. Since the proposed approach provides a unified framework as discussed in Sec. \ref{Sec:Relations}, the following special cases can be taken as the reference for comparison with the existing methods:
\begin{itemize}
	\item $k_{P} = 0$: Direct MRAC with open-loop reference model (Sec. \ref{SubSec:ORM_MRAC})
	\item $k_{P} > 0$, $\mu = 1$: Direct MRAC with closed-loop reference model (Sec. \ref{SubSec:CRM_MRAC})
\end{itemize}

Table \ref{Table:sim_params} summarises the simulation parameters. The simulation code written in \texttt{Julia} can be accessed via \cite{Cho2023}.
\begin{table}[h!tb]
	\begin{center}
		\caption{Simulation Parameters} 
		\label{Table:sim_params}
		\begin{tabular}{ccc}
			\hline
			quantity	& value\\
			\hline
			$x_{0}$, $x_{m_{0}}$, $\hat{\theta}_{0}$	&	$\begin{bmatrix}	0	&	0	&	0\end{bmatrix}^{T}$\\
			$Q_{base}$	&	$\diag\left(\begin{bmatrix}	0 & 0 & 100\end{bmatrix}\right)$\\
			$R_{base}$	&	$1$\\
			$k_{m}^{T}$	&	$\mathsf{lqr}\left(A,b,Q_{base},R_{base}\right)$\\
			$k_{r}$		&	$0$\\
			$\psi\left(x\right)$ & $\frac{1}{2}x^{T}\Gamma^{-1}x$\\ 
			$\Gamma$ & $\diag\left(\begin{bmatrix}	400 & 400 & 20\end{bmatrix}\right)$\\
			$\phi\left(x\right)$ & $0$\\
			$Q$ & $\diag\left(\begin{bmatrix}	1 & 800 & 0.1\end{bmatrix}\right)$\\
			$P$ & $\mathsf{lyap}\left(\left(A-bk_{m}^{T}-k_{P}I\right)^{T}, Q\right)$\\
			$\alpha_{cmd}\left(t\right)$ &	
			$\begin{cases}
				0		& \text{for~} 0		\leq t	<	1\\
				5		& \text{for~} 1		\leq t	<	11\\
				-5		& \text{for~} 11	\leq t	<	22\\
				0		& \text{for~} 22	\leq t	<	41\\
				2.5		& \text{for~} 41	\leq t	<	51\\
				-2.5	& \text{for~} 51	\leq t	<	62\\
				0		& \text{for~} 62	\leq t	< 	80
			\end{cases}$ (in [deg])\\
			$J$	for allocation &	$J_{pert}$ with $W=I$ and $p=2$\\
			$k_{P}(=K_{e})$ & \{$0.0$, $1.0$, $10.0$, $100.0$\}\\
			$\mu$	& \{$0.0$, $0.5$, $1.0$\}\\
			\hline
		\end{tabular}
	\end{center}
\end{table}


\subsection{Simulation Results and Discussion}
Figures \ref{Fig_alpha_coll} and \ref{Fig_q_coll} show the state responses for both the plant and the virtual dynamics. Figure \ref{Fig_u_coll} shows the time histories of the input along with each of its components, and Fig. \ref{Fig_Delta_aug_coll} shows the time histories of the true and approximated uncertainty. Also, Figs. \ref{Fig_V_e_coll}-\ref{Fig_V_dtheta_coll} show the time histories of the $2$-norm of vector quantities, namely, the tracking error, the time-derivative of tracking error, the parameter estimation error, and the time-derivative of estimated parameter, respectively. Each subfigure in Figs. \ref{Fig_alpha_coll}-\ref{Fig_V_dtheta_coll} corresponds to the combination $\left(k_{P},\mu\right)$ of its value written above.

\subsubsection{Alleviation of High-Frequency Oscillation}
Increasing $k_{P}$ smoothes the transients in the initial phase, leading to the reduced peak value in the time histories of $\left\|\dot{\hat{\theta}}\left(t\right)\right\|$ as shown in Fig. \ref{Fig_V_dtheta_coll}. Also, increasing $k_{P}$ reduces the high-frequency oscillation in $q\left(t\right)$, $\delta_{e}\left(t\right)$, $e\left(t\right)$, $\dot{e}\left(t\right)$, $\tilde{\theta}\left(t\right)$, and $\dot{\hat{\theta}}\left(t\right)$ as shown in Figs. \ref{Fig_q_coll}, \ref{Fig_u_coll}, \ref{Fig_V_e_coll}, \ref{Fig_V_de_coll}, \ref{Fig_V_theta_coll}, and \ref{Fig_V_dtheta_coll}, respectively.

\subsubsection{Alleviation of Peaking Phenomenon}
Given a fixed $\Gamma$, increasing $k_{P}$ with $\mu=1$ causes the low-frequency high-amplitude oscillation called the peaking phenomenon as shown in Figs. \ref{Fig_alpha_coll} and \ref{Fig_q_coll}. The unwanted behaviour has prevented the existing CRM-MRAC method from increasing the coupling gain $k_{P}$ above certain value. The proposed approach can alleviate the peaking response by decreasing $\mu$ towards zero even if $k_{P}$ is large as shown in Figs. \ref{Fig_alpha_coll} and \ref{Fig_q_coll}, at the expense of less model learning taking place using direct adaptation as discussed below.

\subsubsection{Model Learning Capability of Direct Adaptation Law}
Since the proposed approach prescribes the same dynamics of $e\left(t\right)$ for a given $k_{P}$ regardless of $\mu$, Figs. \ref{Fig_V_e_coll} and \ref{Fig_V_theta_coll} show that the overall pattern of $\left\|e\left(t\right)\right\|$ and $\left\|\tilde{\theta}\left(t\right)\right\|$ is mostly dependent on $k_{P}$. Varying $\mu$ contributes to the slight difference in $\left\|\tilde{\theta}\left(t\right)\right\|$ through the change in $\Phi\left(x\left(t\right)\right)$ due to the varied state responses. 

The direct adaptation law with a fixed $\Gamma$ tends to lose its model learning capability with larger $k_{P}$ as shown in Figs. \ref{Fig_Delta_aug_coll} and \ref{Fig_V_theta_coll}. This can be understood as a consequence of faster decay of $e\left(t\right)$ to the neighbourhood of zero and less rich $\Phi\left(x\left(t\right)\right)$. In such cases, direct adaptation alone is not effective for learning even though the amplitude of oscillation due to the peaking phenomenon can grow very large. The case of $k_{P}=10^{3}$ showed $\dot{\hat{\theta}}\left(t\right) \approx 0$, leading to almost complete silence in $\tilde{\theta}\left(t\right)$. One may need to use a larger $\Gamma$ for tighter steady-state tracking to compensate for the effect of insufficient learning. The results indicate that using a high coupling gain $k_{P}$ with $\mu=0$ tends to isolate the online model learning task from stable tracking control. This suggests the necessity to implement a composite adaptation algorithm when long-term learning of model is desired.

\subsubsection{Blending of Uncertainty Cancallation Strategies}
Overall, increasing $k_{P}$ with the choice $\mu = 0$ leads to the uncertainty cancellation behaviour that incorporates instantaneous uncertainty rejection alongside model learning. The case of $\left(k_{P},\mu\right)=\left(100,0\right)$ in Fig. \ref{Fig_Delta_aug_coll} shows that the model prediction output $-u_{ad}\left(t\right)=\Phi\left(x\left(t\right)\right)^{T}\hat{\theta}\left(t\right)$ alone exhibits a large gap from the true uncertainty $\Delta\left(t\right)$ while the sum of two input components $-u_{ad}\left(t\right)-u_{c}\left(t\right)$ is close to $\Delta\left(t\right)$. The result implies that the fast element $u_{c}\left(t\right)$ contributes more to the control input in the absence of sufficient learned estimate, thus compensating for the lack of learning. In a similar manner to $\mathcal{L}_{1}$-$\mathcal{GP}$ in \cite{Gahlawat2020}, the relative contribution of the learned model output will increase as the model becomes more accurate over time when using a composite adaptation law.

To understand the physical reasons for the observed trend, let us assume the presence of the observer $\Sigma_{o}$ which is designed to have $A_{o}=A_{m}-k_{P}I$. One can easily show that the discrepancy between $x_{m}\left(t\right)$ and $\hat{x}\left(t\right)$ follows $\dot{\tilde{x}}_{m}=\left(A_{m}-k_{P}I\right)\tilde{x}_{m}+b\Phi^{T}\hat{\theta}$ where $\tilde{x}_{m}\left(t\right):=x_{m}\left(t\right)-\hat{x}\left(t\right)=e\left(t\right)-\tilde{x}\left(t\right)$. Then, assuming the boundedness of $\Phi\left(t\right)$, increasing $k_{P}$ results in the decrease of the uniform bound over $\tilde{x}_{m}\left(t\right)$ as $\hat{\theta}\left(t\right)$ is bounded as discussed in Sec. \ref{SubSec:ML}. Therefore, the behaviour of $\Sigma_{m}$ tends to be closer to the pure observer $\Sigma_{o}$ used for the instantaneous rejection of lumped uncertainty with a larger $k_{P}$. 

In summary, the proposed approach tends to adjust its uncertainty cancellation behaviour depending on $\mu$. The particular choice $\mu=0$ tends to blend instantaneous rejection of lumped uncertainty complementarily with online learning of uncertainty model for uncertainty cancellation. 

\subsubsection{Practically Useful Coupling Input Allocation Strategy}
Setting $W=I$, $\mu=0$, and $p=2$ leads to the allocation strategy minimising $J_{pert}=\left\|U_{m}\left(t\right)\right\|_{2}$ at each $t$ subject to the given coupling input realisation constraint. This allocation keeps the same nominal dynamics in $e\left(t\right)$ while minimising the discrepancy between the open-loop reference model $\Sigma_{id}$ in Eq. \eqref{Eq:sys_id} and the closed-loop reference model $\Sigma_{m}$ in Eq. \eqref{Eq:sys_m}. The physical meaning of selecting $\mu = 0$ is to maintain the benefits of introducing distinction between the times scales of $\Sigma_{e}$ and $\Sigma_{z}$ (which is equivalent to $\Sigma_{m}$ when $\mu =0$) in improving the transient response characteristics while alleviating the undesirable consequences of having a feedback-interconnected reference model such as the peaking phenomenon that occurs when $k_{P}$ is large.

\begin{figure}[!htbp]
	\begin{center}
		\includegraphics[width=\textwidth]{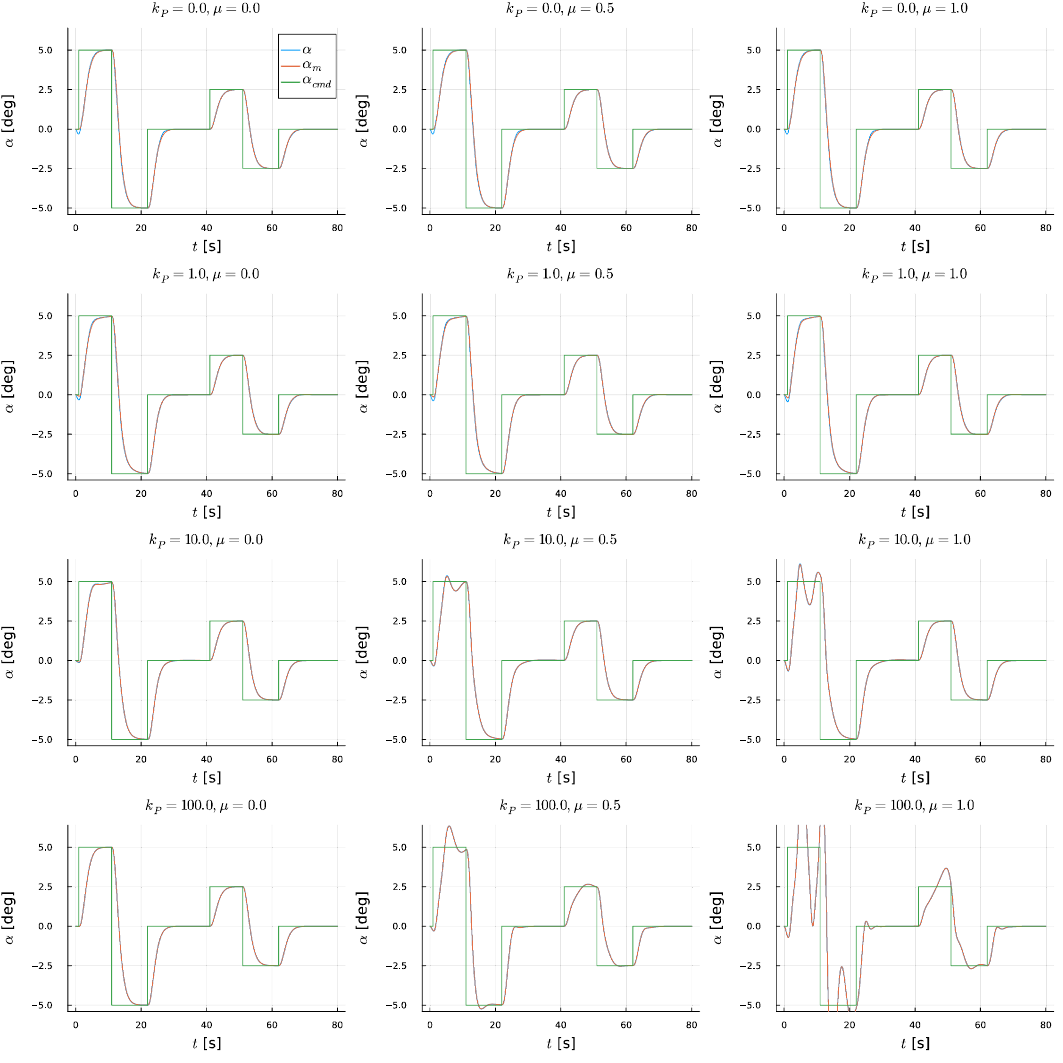}
		\caption{Angle-of-Attack $\alpha$}
		\label{Fig_alpha_coll}
	\end{center}
\end{figure}

\begin{figure}[!htbp]
	\begin{center}
		\includegraphics[width=\textwidth]{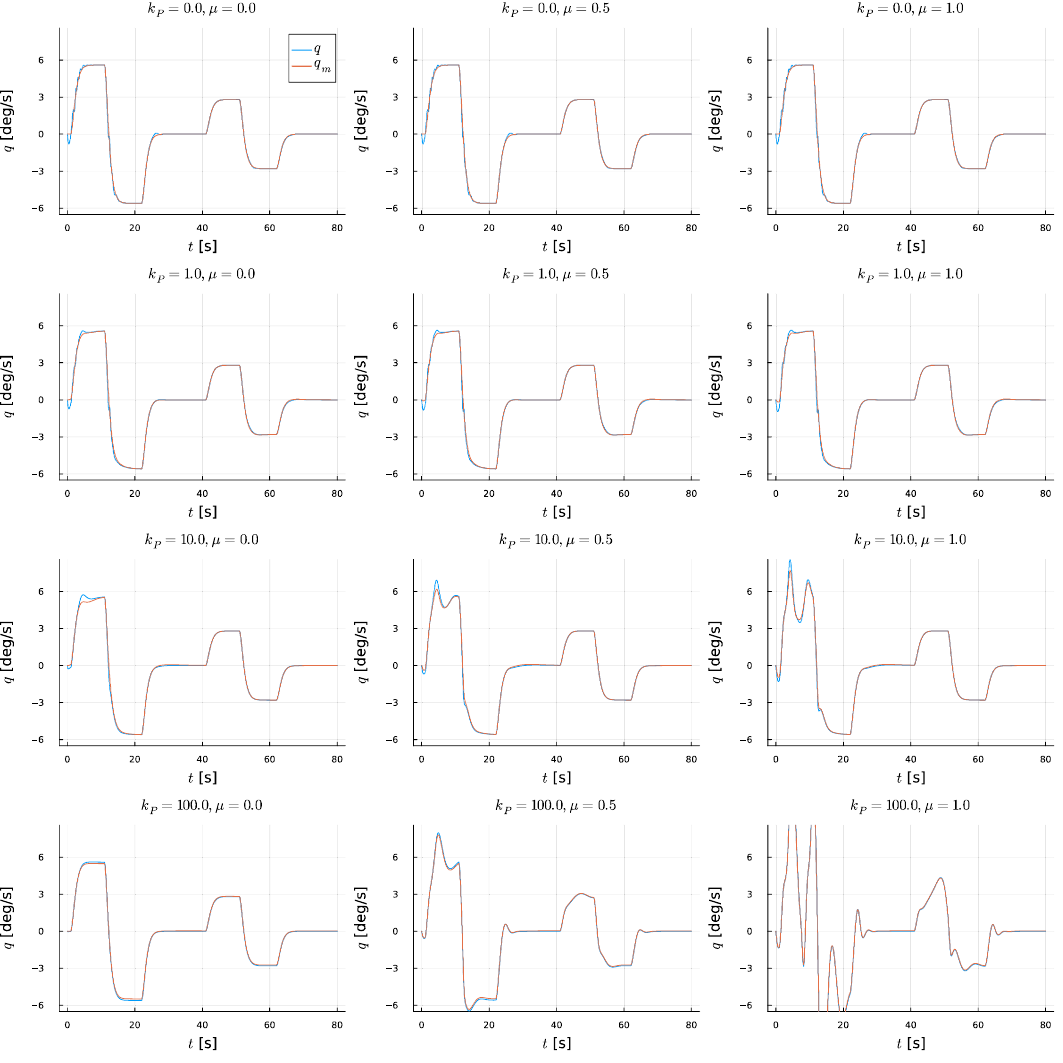}
		\caption{Pitch Rate $q$}
		\label{Fig_q_coll}
	\end{center}
\end{figure}

\begin{figure}[!htbp]
	\begin{center}
		\includegraphics[width=\textwidth]{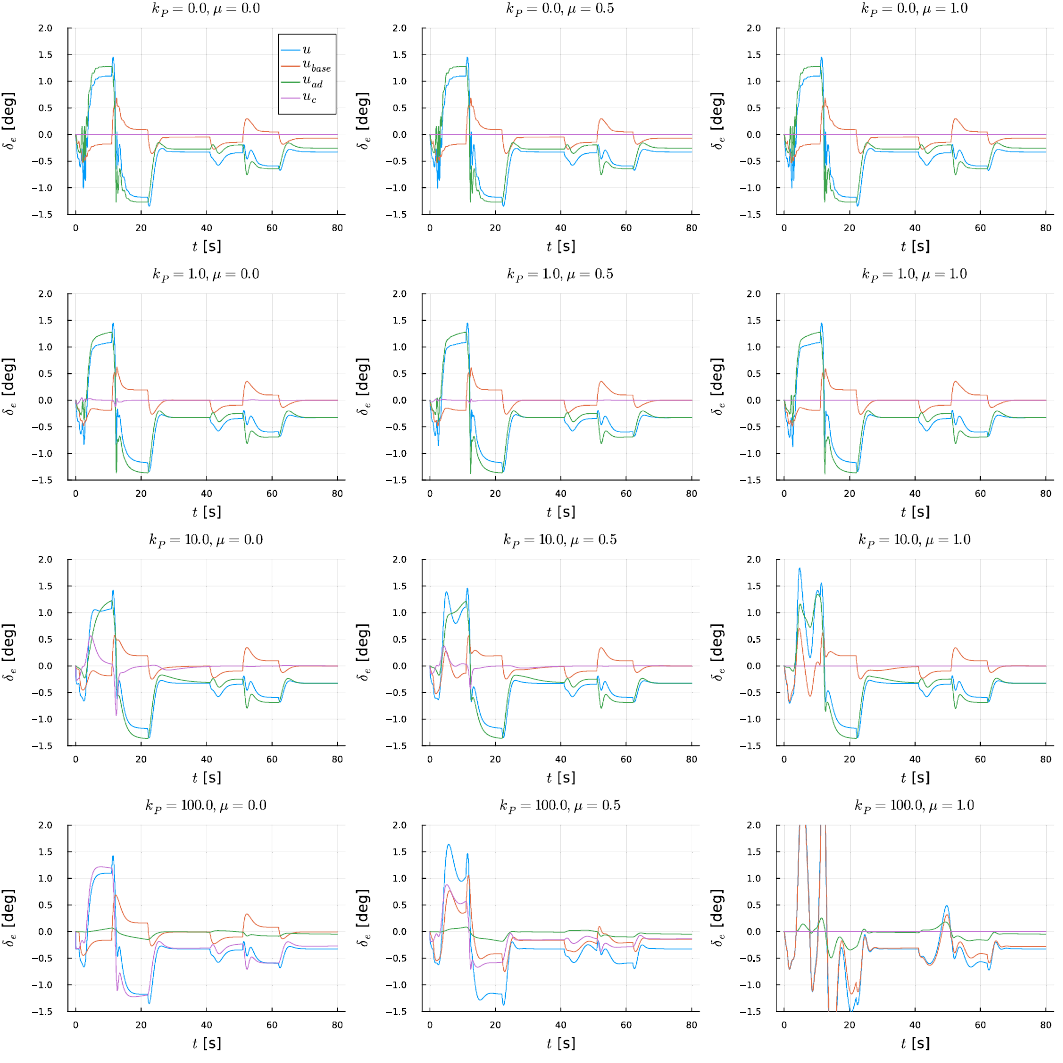}
		\caption{Elevator Deflection $\delta_{e}$}
		\label{Fig_u_coll}
	\end{center}
\end{figure}

\begin{figure}[!htbp]
	\begin{center}
		\includegraphics[width=\textwidth]{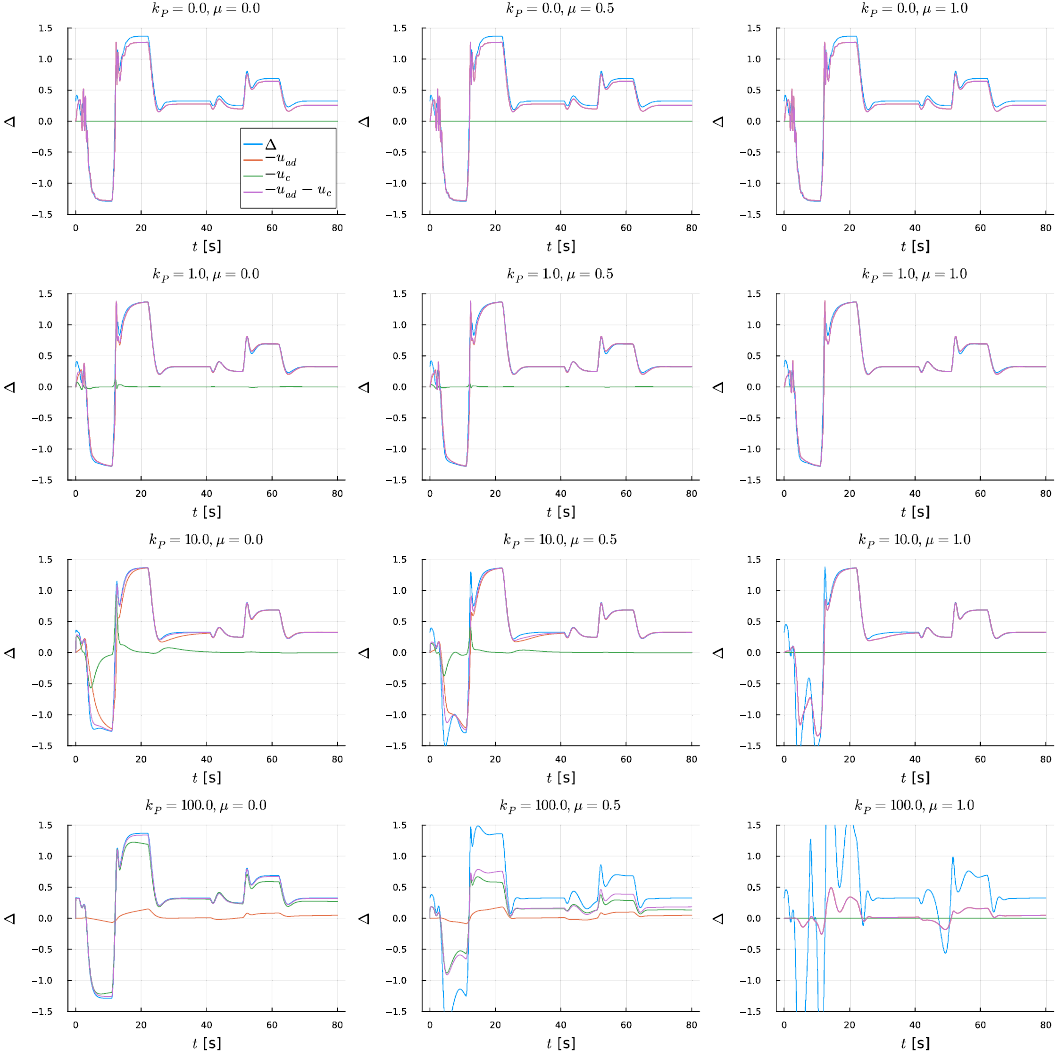}
		\caption{True Uncertainty $\Delta$ and Approximated Uncertainty}
		\label{Fig_Delta_aug_coll}
	\end{center}
\end{figure}

\begin{figure}[!htbp]
	\begin{center}
		\includegraphics[width=\textwidth]{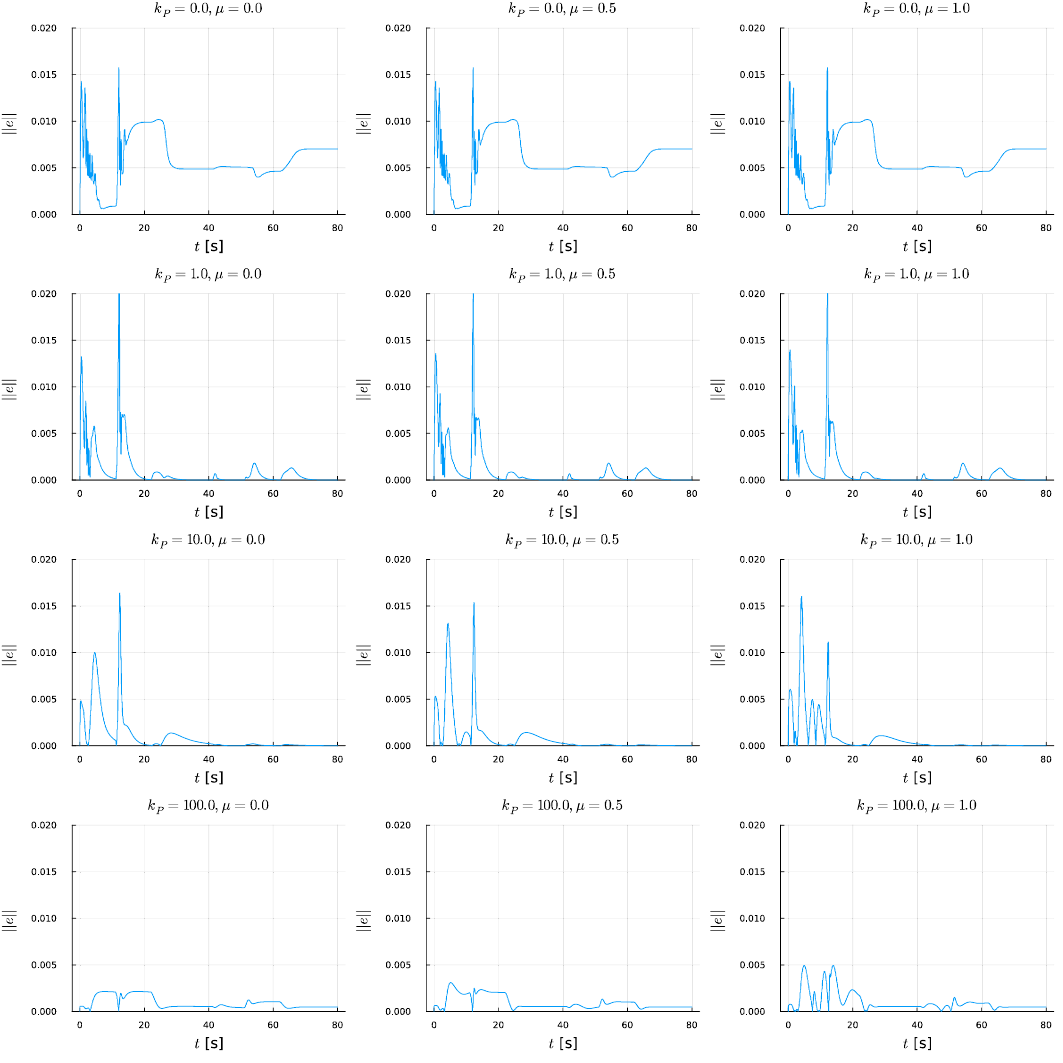}
		\caption{Norm of Tracking Error $\left\|e\right\|$}
		\label{Fig_V_e_coll}
	\end{center}
\end{figure}

\begin{figure}[!htbp]
	\begin{center}
		\includegraphics[width=\textwidth]{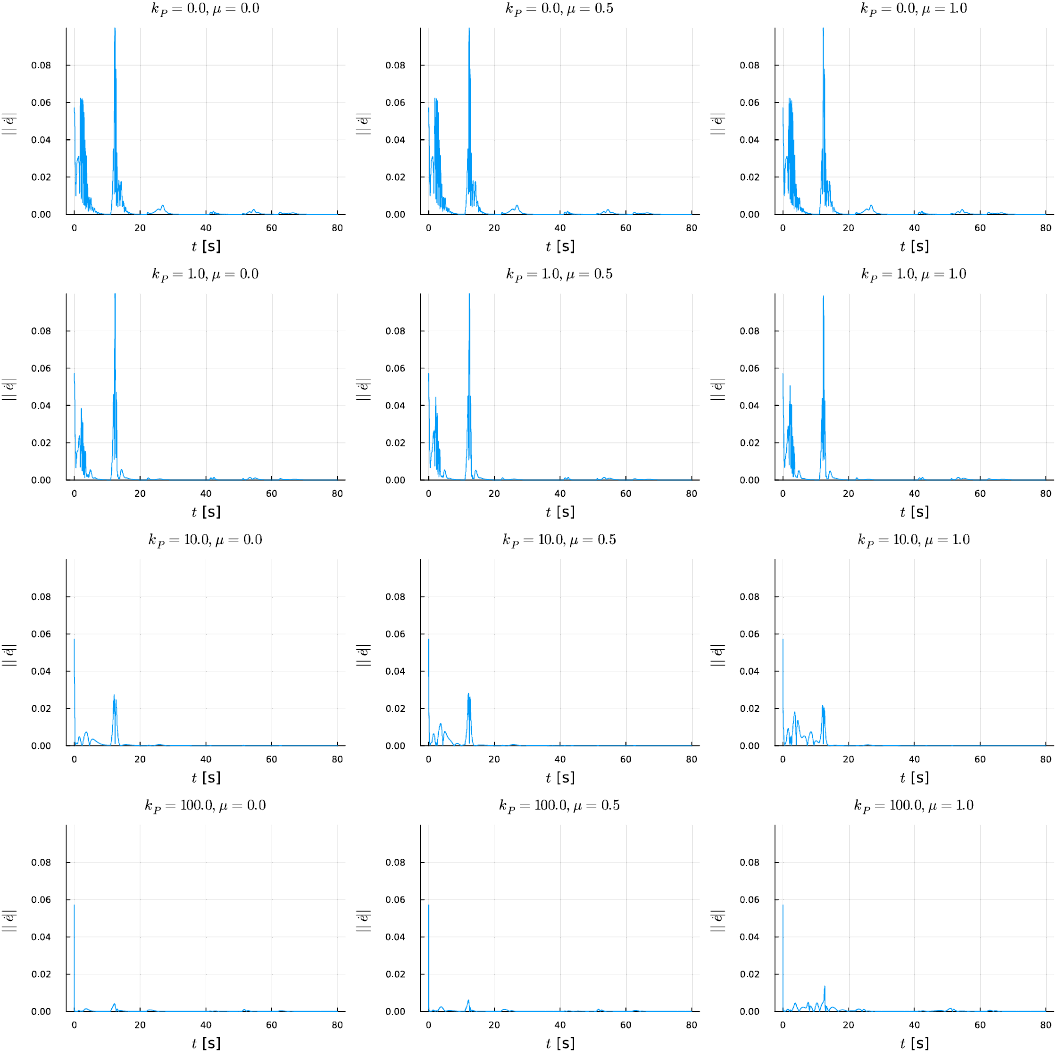}
		\caption{Norm of Time-Derivative of Tracking Error $\left\|\dot{e}\right\|$}
		\label{Fig_V_de_coll}
	\end{center}
\end{figure}

\begin{figure}[!htbp]
	\begin{center}
		\includegraphics[width=\textwidth]{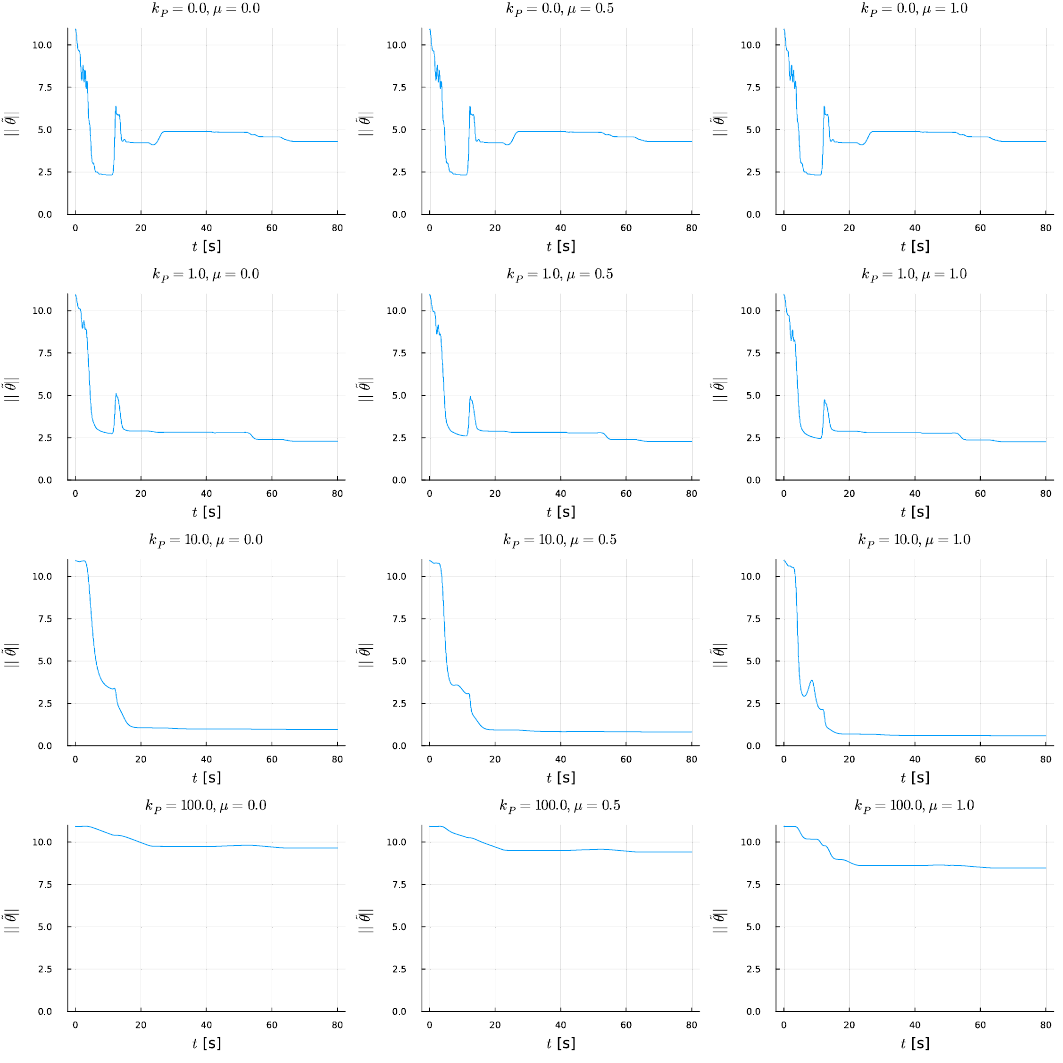}
		\caption{Norm of Parameter Estimation Error $\left\|\tilde{\theta}\right\|$}
		\label{Fig_V_theta_coll}
	\end{center}
\end{figure}

\begin{figure}[!htbp]
	\begin{center}
		\includegraphics[width=\textwidth]{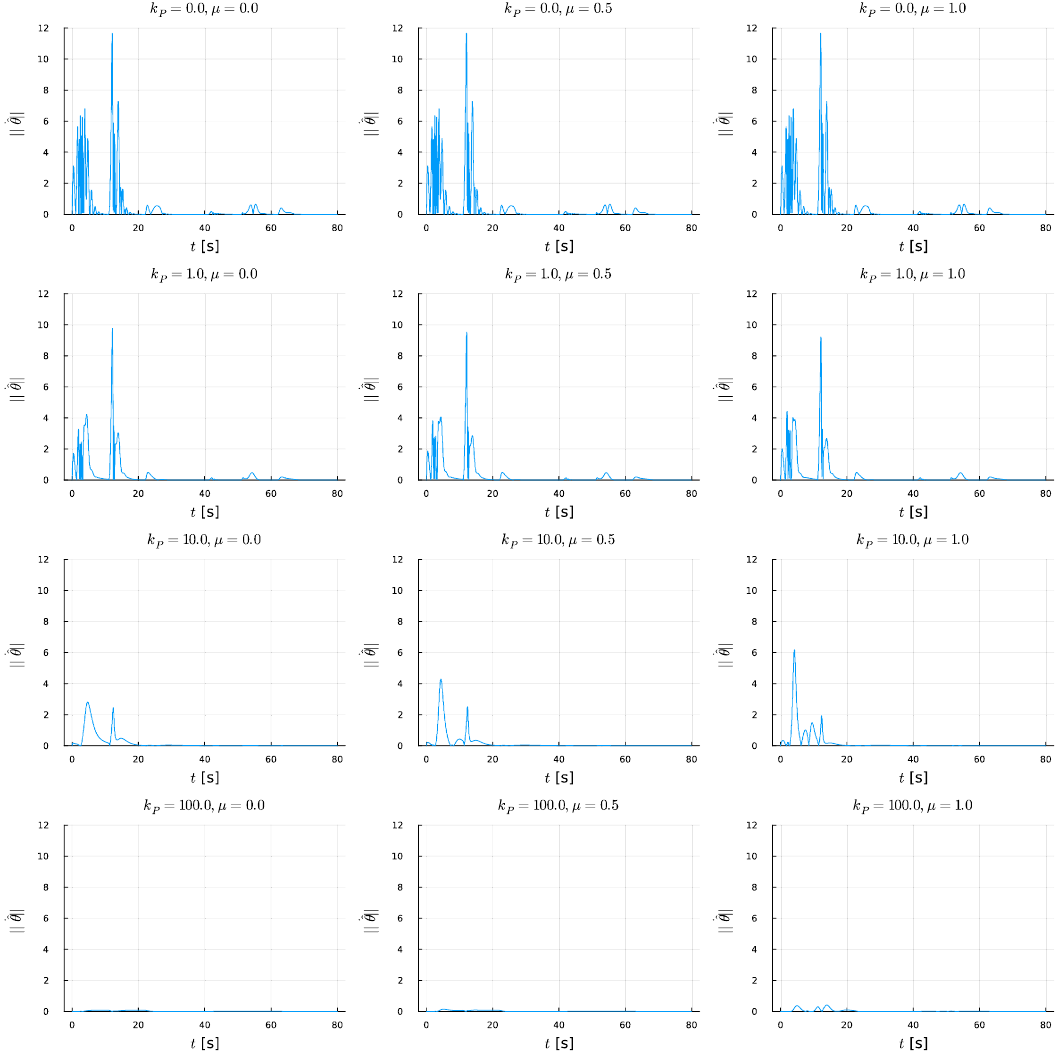}
		\caption{Norm of Time-Derivative of Estimated Parameter $\left\|\dot{\hat{\theta}}\right\|$}
		\label{Fig_V_dtheta_coll}
	\end{center}
\end{figure}

\section{Conclusions} \label{Sec:Concls}
Synchronisation of the physical plant and the virtual dynamic system was suggested as a design approach for adaptive control systems that natively addresses loop shaping goals as well as desired collective behaviour generation. The proposed approach introduces the design as well as allocation of the coupling input as the means of systematically adjusting the response characteristics. Regarding the design part, time-scale separation between the dynamics for synchronisation process and that for collective behaviour is central to alleviation of high-frequency oscillatory transients. Regarding the allocation part, reducing the contribution of the virtual dynamics input in synchronisation is central to alleviation of the undesired peaking phenomenon. The allocation minimising the virtual dynamics input allows complementary blending of the uncertainty cancellation behaviours between instantaneous rejection and online model learning according to the learning progress. This can be useful to keep consistent performance during the initial learning transients in the presence of unlearned uncertainties.

\appendix
\section*{Appendix: Uncertainty Approximation in Adaptive Control} \label{Sec:Delta_hat}
\renewcommand{\thesubsection}{\Alph{subsection}}
\setcounter{equation}{0}
\renewcommand{\theequation}{\thesubsection.\arabic{equation}}
\setcounter{REM}{0}
\renewcommand{\therem}{\thesubsection.\arabic{REM}}

In general, two different strategies are available for uncertainty approximation; i) instantaneous rejection of lumped uncertainty (Appendix \ref{SubSec:DOB}), and ii) online learning of a model (Appendix \ref{SubSec:ML}). Both approaches usually rely on a recursive state/disturbance observer since the time-derivative of plant state $\dot{x}$ cannot be measured in practice. The former approach applies the observer output directly as the input component $u_{ad}\left(t\right)$ for uncertainty cancellation. The latter approach takes the observer output as the target data, performs real-time regression using the streamed data sequence, and applies the predicted output of the learned model as the input component $u_{ad}\left(t\right)$ for uncertainty cancellation. 

\subsection{Instantaneous Rejection of Lumped Uncertainty} \label{SubSec:DOB}
\subsubsection{Observer}
Let us frame an observer to be a certainty-equivalent copy of the plant with an observation error feedback term as
\begin{equation} \label{Eq:sys_o}
	\Sigma_{o}: \qquad 
		\dot{\hat{x}}\left(t\right) = Ax\left(t\right) + bu\left(t\right) + A_{o}\tilde{x}\left(t\right), \quad
		\hat{x}\left(0\right) = \hat{x}_{0}
\end{equation}
where $\hat{x}\left(t\right)$ denotes the observer state, $\tilde{x}\left(t\right):= \hat{x}\left(t\right) - x\left(t\right)$ denotes the observation error, and $A_{o}$ is a Hurwitz matrix which is a design parameter. Then, the observation error dynamics can be obtained from Eqs. \eqref{Eq:sys_p} and \eqref{Eq:sys_o} as 
\begin{equation} \label{Eq:sys_f}
	\Sigma_{f} = \Sigma_{o} - \Sigma_{p}: \qquad 
		\dot{\tilde{x}}\left(t\right) = A_{o}\tilde{x}\left(t\right) - b\Delta\left(t\right), \quad
		\tilde{x}\left(0\right) = \hat{x}_{0}-x_{0}
\end{equation}
The evolution of $\tilde{x}\left(t\right)$ over time following Eq. \eqref{Eq:sys_f} is driven by the unknown quantity $\Delta\left(t\right)$, however, $\tilde{x}\left(t\right)$ is measurable as it can be obtained by computing the difference between $\hat{x}\left(t\right)$ and $x\left(t\right)$. 

The system $\Sigma_{f}$ in Eq. \eqref{Eq:sys_f} describes a low-pass filter for which bandwidth is determined by $A_{o}$. The simplest design is to take $A_{o} = -\omega_{f}I$ with a scalar constant $\omega_{f} > 0$. In this case, the observation error is related to the uncertainty as
\begin{equation} \label{Eq:tilde_x}
	\tilde{x}\left(t\right) = e^{-\omega_{f}t}\tilde{x}\left(0\right) - b\Delta_{f}\left(t\right)
\end{equation}
where the filtered uncertainty $\Delta_{f}\left(t\right)$ is given by
\begin{equation} \label{Eq:Delta_f}
	\Delta_{f}\left(t\right) = \int_{0}^{t}e^{-\omega_{f}\left(t-\tau\right)}\Delta\left(\tau\right)d\tau = C\left(s\right)\Delta\left(t\right)
\end{equation}
with $C\left(s\right) = \frac{\omega_{f}}{s+\omega_{f}}$. 

The instantaneous rejection approach is to set $\hat{\Delta}\left(t\right)$ in Eq. \eqref{Eq:u_ad} to be equal to $\Delta_{f}\left(t\right)$. It is impossible to measure $\Delta_{f}\left(t\right)$ directly, however, $b\Delta_{f}\left(t\right) = e^{-\omega_{f}t}\tilde{x}\left(0\right)-\tilde{x}\left(t\right) $ is a known signal since the right-hand-side can be computed. One can then solve the linear equation and get
\begin{equation} \label{Eq:Delta_hat_DOB}
	\hat{\Delta}\left(t\right) = b^{\dagger}\left\{ e^{-\omega_{f}t}\tilde{x}\left(0\right) -\tilde{x}\left(t\right) \right\}
\end{equation}
where $b^{\dagger} = \frac{b^{T}}{b^{T}b}$ for a vector $b \neq 0$. 

\subsubsection{Design Considerations}
Besides the separation between the bandwidth $\omega_{e}$ of $\Sigma_{e}^{CL}$ and the bandwidth $\omega_{f}$ of $\Sigma_{f}$ (discussed in the end of Sec. \ref{SubSubSec:Sigma_e_CL} and in the beginning of Sec. \ref{Sec:Delta_hat}), the choice of $\omega_{f}$ is also limited by the fact that the uncertainty lying beyond the actuator bandwidth $\omega_{a}$ cannot be cancelled by matching a plant input. The input component for uncertainty cancellation $u_{ad}\left(t\right) = -\hat{\Delta}\left(t\right)$ acts only through the physical actuator of the plant with a finite bandwidth. Hence, $\omega_{a}$ is in practice the maximum admissible value of $\omega_{f}$ for meaningful uncertainty cancellation. The same reasoning is the underlying design principle of the $\mathcal{L}_{1}$ adaptive control methods. In summary, a reasonable design guideline is to choose parameters to satisfy $\max \left|\lambda\left(A_{m}\right)\right| \ll \omega_{e} < \omega_{f} \leq \omega_{a}$.

\begin{remapp}
One may introduce an adaptive input component $U_{ad}\left(t\right)$ in the virtual system, in addition to the coupling input component $U_{m}\left(t\right)$, to cancel the remaining uncertainty component that could not be cancelled out by $u_{ad}\left(t\right)$ of the actual plant from $\Sigma_{e}$. Such examples include the case of unmatched uncertainty and the high-frequency uncertainty lying above $\omega_{a}$. Distributing the input of $\Sigma_{e}$ to both $\Sigma_{p}$ and $\Sigma_{m}$ not only for coupling but also for uncertainty cancellation might result in tighter synchronisation by reducing the discrepancy between $\Sigma_{e}$ and its ideal counterpart. However, it does not necessarily imply better performance characteristics in the plant state response in that the adaptive input component in the virtual system cannot cancel any uncertainty from $\Sigma_{p}^{CL}$ even though it has an effect in $\Sigma_{e}$.
\end{remapp}

\subsection{Online Learning of Uncertainty Model} \label{SubSec:ML}
\subsubsection{Overview}
The synchronisation approach prioritises stable synchronisation and separates out model learning in the sense of real-time identification as a secondary task, provided that the uncertainty approximation algorithm and the coupling input are designed respecting the separation between the time scales of associated dynamic systems. This separation is a natural architectural choice intended in the synchronisation-oriented approach rather than a result specific to a certain uncertainty model learning algorithm. 

Nevertheless, the behaviour of synchronisation dynamics $\Sigma_{e}^{CL}$ is not completely agnostic to the learning dynamics, i.e., adaptation algorithm, when using the predicted output of online-learned model in the controller. The model learning algorithm should be designed to guarantee stability of the overall feedback interconnected system, as well as asymptotic convergence of $e$ to zero in the absence of inevitable learning residual due to finite learning capability (or stronger boundedness properties in the presence of the learning residual).

Online model learning can improve transient characteristics essentially by removing the uncertainty approximation residual $\tilde{\Delta}\left(t\right)$ over time, as long as the feedback interconnection of the closed-loop synchronisation error dynamics and the model learning dynamics is stable as a whole. Constructing the signal $\hat{\Delta}\left(t\right)$ based on gradual model learning instead of instantaneous disturbance rejection for use in control can be formulated as the minimisation of a composite loss function consisting of a synchronisation loss and a regression loss.


\subsubsection{Online Model Learning as Parametric Regression}
Various regression algorithms can be adopted for online learning depending on the structure chosen/given for the uncertainty model \cite{Chowdhary2015, Joshi2019, Joshi2020, Boffi2022}. 
To avoid diverging from the main point, this study proceeds with a linear parametric model for the uncertainty that can be expressed as
\begin{equation} \label{Eq:Delta_linear}
	\Delta\left(t\right) = \Phi\left(x\left(t\right)\right)^{T}\theta
\end{equation}
where $\theta\in\mathbb{R}^{p}$ represents the uncertain constant parameter, and $\Phi\left(x\left(t\right)\right) \in\mathbb{R}^{p}$ represents the known state-dependent basis function which is also called the feature. For this model, the estimated uncertainty can be written as
\begin{equation} \label{Eq:Delta_hat_ML}
	\hat{\Delta}\left(t\right) = \Phi\left(x\left(t\right)\right)^{T}\hat{\theta}\left(t\right) 
\end{equation}
where $\hat{\theta}\left(t\right)$ represents the estimated parameter. 

The goal of online model learning is to achieve stable convergence of the parameter estimation error $\tilde{\theta}\left(t\right):= \hat{\theta}\left(t\right) - \theta$ to zero. It should be distinguished from the weaker goal of nullifying $\tilde{\Delta}\left(t\right)=\Phi\left(x\left(t\right)\right)^{T}\tilde{\theta}\left(t\right)$, since $\tilde{\theta}\left(t\right)$ can have nonzero component in the direction orthogonal to $\Phi\left(x\left(t\right)\right)$ even when $\tilde{\Delta}\left(t\right)$ is zero.

Model learning task can be formulated as minimisation of a loss function $L\left(Y\left(t;\theta\right),\hat{Y}\left(t;\hat{\theta}\left(t\right)\right)\right)$ where $Y\left(t;\theta\right)$ is a measurable target signal carrying information about $\theta$ and $\hat{Y}\left(t;\hat{\theta}\left(t\right)\right)$ is a predicted output given by a function of $\hat{\theta}\left(t\right)$. Unlike the Lyapunov function for analysis purpose, the loss function should be a function of known quantities to update the parameter by evaluating the local geometry of loss landscape. 

\subsubsection{Feature Construction Over Time}
In the continuous-time data streaming setup, the target signal $Y\left(t;\theta\right)$ can be \emph{constructed} by designing i) a regressor filter resembling the observation error system $\Sigma_{f}$, and ii) a feature extender which internally constructs and retains a mini-batch-like memory of the information extracted from the observed data signals. 

Let us define the regressor filter that generates $\Phi_{f}\left(t\right) \in \mathbb{R}^{p \times n}$ as
\begin{equation} \label{Eq:sys_Phif}
	\Sigma_{\Phi f}: \qquad 
		\dot{\Phi}_{f}^{T}\left(t\right) = A_{o}\Phi_{f}^{T}\left(t\right)-b\Phi\left(x\left(t\right)\right)^{T}, \quad
		\Phi_{f}\left(0\right) = 0
\end{equation}
with $A_{o}$ chosen for $\Sigma_{o}$ in Eq. \eqref{Eq:sys_o}. By substituting Eq. \eqref{Eq:Delta_linear} into Eq. \eqref{Eq:sys_f}, the observation error can be represented as
\begin{equation} \label{Eq:Phif_theta}
	\tilde{x}\left(t\right) = \Phi_{f}^{T}\left(t\right)\theta
\end{equation}
when the initial condition is given by $\tilde{x}\left(0\right)=0$. Then, the feature extender for constructing the information matrix $\Omega\left(t\right) \in \mathbb{R}^{q \times p}$ for some $q \geq 1$ and the target signal $Y\left(t;\theta\right)$ can be designed as a stable and time-varying filter that applies to the output of regressor filter and evolves according to
\begin{equation} \label{Eq:sys_Phie}
	\Sigma_{\Phi e}:~~~
	\begin{aligned}
		\dot{\Omega}\left(t\right) &= A_{Y}\left(t\right)\Omega\left(t\right)+B_{Y}\left(t\right)\Phi_{f}^{T}\left(t\right), &\Omega\left(0\right) = 0\\
		\dot{\eta}\left(t\right) &= A_{Y}\left(t\right)\eta\left(t\right) + B_{Y}\left(t\right)\tilde{x}\left(t\right), &\eta\left(0\right) = 0\\
		Y\left(t;\theta\right) &= \mathcal{M}\left[\eta\left(t\right)\right]
	\end{aligned}
\end{equation}
where $A_{Y}\left(t\right)\in\mathbb{R}^{q\times q}$ is a forgetting factor such that $A_{Y}\left(t\right) + A_{Y}^{T}\left(t\right)\leq 0$ for stability, $B_{Y}\left(t\right)\in\mathbb{R}^{q\times n}$ is an update factor or an exogenous driving signal, 
and $\mathcal{M}\left[\cdot\right]$ is a linear operator which generates a bounded output in the form of either \emph{square or tall} matrix. By substituting Eq. \eqref{Eq:Phif_theta} into Eq. \eqref{Eq:sys_Phie}, we have 
\begin{equation} \label{Eq:eta_theta}
	\eta\left(t\right) = \Omega\left(t\right) \theta 
\end{equation}
and therefore, the target signal can be represented as
\begin{equation} \label{Eq:Y}
	Y\left(t;\theta\right) = \mathcal{M}\left[\Omega\left(t\right)\right]\theta
\end{equation}
Correspondingly, the predicted output can be represented as
\begin{equation} \label{Eq:Y_hat}
	\hat{Y}\left(t;\hat{\theta}\left(t\right)\right) = \mathcal{M}\left[\Omega\left(t\right)\right]\hat{\theta}\left(t\right)
\end{equation}

\begin{remapp}
Conceptually, the design of the feature extender aims to achieve full rank and to condition the matrix $\mathcal{M}\left[\Omega\left(t\right)\right]$ in a short interval of time so that the unknown parameter $\theta$ becomes identifiable from the known signals $Y\left(t;\theta\right)$ and $\mathcal{M}\left[\Omega\left(t\right)\right]$. One possible strategy is to perform pointwise-in-time optimisation of the forgetting and update factors considering the change of singular value/vector structure (or eigenstructure) of $\mathcal{M}\left[\Omega\left(t\right)\right]$. 
\end{remapp}

\subsubsection{Design of Parameter Update Law for Stable Adaptation} \label{SubSubSec:stable_adapt}
The task is to design the update law that specifies $\dot{\hat{\theta}}\left(t\right)$ to achieve stable convergence of both errors $\mathbf{e}_{I}^{l}\left(t\right)$ and $\tilde{\theta}\left(t\right)$ to zero. Continuous-time recursive algorithms based on gradient flow and its variants can be applied for minimisation of the loss function using the information of local loss geometry evaluated at $\hat{\theta}\left(t\right)$, such as $L$, $\nabla_{\hat{\theta}}L$, and $\nabla_{\hat{\theta}}^{2}L$.

Let $V_{e}\left(\mathbf{e}_{I}^{l}\left(t\right) \right)= \frac{1}{2}{\mathbf{e}_{I}^{l}}^{T}\left(t\right)P\mathbf{e}_{I}^{l}\left(t\right)$ be the Lyapunov function certifying global exponential stability of $\Sigma_{e}^{CL}$'s origin under the stabilising coupling input $U_{c}\left(\mathbf{e}_{I}^{l}\left(t\right) \right)=-K_{e}\mathbf{e}_{I}^{l}\left(t\right)$ \emph{in the absence of uncertainty} where $P>0$ is the unique solution satisfying $\left(A_{e}-B_{e}K_{e}\right)^{T}P + P\left(A_{e}-B_{e}K_{e}\right) + Q =0$ for any given $Q>0$. Then, the Lyapunov function verifies
\begin{equation} \label{Eq:V_e_dot_nominal}
	\begin{aligned}
	\left.\dot{V}_{e}\right|_{\Delta=0} &= \nabla V_{e} \left.\dot{\mathbf{e}}_{I}^{l}\right|_{\Delta=0}
	={\mathbf{e}_{I}^{l}}^{T}P\left(A_{e}-B_{e}K_{e}\right)\mathbf{e}_{I}^{l} =\frac{1}{2}{\mathbf{e}_{I}^{l}}^{T} \left\{\left(A_{e}-B_{e}K_{e}\right)^{T}P + P\left(A_{e}-B_{e}K_{e}\right)\right\}\mathbf{e}_{I}^{l}\\
	& = -\frac{1}{2}{\mathbf{e}_{I}^{l}}^{T}Q\mathbf{e}_{I}^{l} < 0
	\end{aligned}
\end{equation}

Now, consider a Lyapunov function $V$ for certifying stability and convergence properties of the overall closed-loop adaptive control system \emph{in the presence of uncertainty}. Suppose that the Lyapunov function has an additive form $V = V_{e}\left(\mathbf{e}_{I}^{l}\left(t\right) \right) + V_{\tilde{\theta}}\left(\tilde{\theta}\left(t\right)\right)$ where $V_{\tilde{\theta}}$ is the component added in association with model learning. Let us choose 
\begin{equation} \label{Eq:V_theta}
	V_{\tilde{\theta}} = D_{\psi}\left(\theta, \hat{\theta} \right) = \psi\left(\theta\right) - \psi\left(\hat{\theta}\right) +\nabla\psi\left(\hat{\theta}\right)\tilde{\theta}
\end{equation}
where $D_{\psi}$ is the Bregman divergence associated with a continuously-differentiable and strictly convex function $\psi\left(\cdot\right)$. Then, $V_{\tilde{\theta}} \geq 0$ for all $\theta$ and $\hat{\theta}$ as a consequence of the convexity of $\psi\left(\cdot\right)$ and $V_{\tilde{\theta}}=0$ if and only if $\tilde{\theta}=0$ due to the strict convexity of $\psi\left(\cdot\right)$.

In general, we require the time-derivative of Lyapunov function along the trajectory of uncertain system to be at least negative semi-definite to guarantee Lyapunov stability. By definition, we have $\dot{V}_{\tilde{\theta}} = \frac{d}{dt}\nabla\psi\left(\hat{\theta}\right)\tilde{\theta}= \tilde{\theta}^{T}\nabla^{2}\psi\left(\hat{\theta}\right)\dot{\hat{\theta}}$. Therefore,
\begin{equation} \label{Eq:V_dot}
		\dot{V} = \dot{V}_{e} + \dot{V}_{\tilde{\theta}}
			= \left.\dot{V}_{e}\right|_{\Delta=0} + \frac{\partial V_{e}}{\partial e}b\Phi^{T}\tilde{\theta} + \dot{V}_{\tilde{\theta}}
			=-\frac{1}{2}{\mathbf{e}_{I}^{l}}^{T}Q\mathbf{e}_{I}^{l} + {\mathbf{e}_{I}^{l}}^{T}PB_{e}b\Phi^{T}\tilde{\theta} +  \tilde{\theta}^{T}\nabla^{2}\psi\left(\hat{\theta}\right)\dot{\hat{\theta}}
\end{equation}
If we choose the regression loss function also in the form of a Bregman divergence $L\left(Y,\hat{Y}\right):=D_{\phi}\left(Y,\hat{Y}\right)$ with a twice differentiable strictly convex potential function $\phi\left(\cdot\right)$, then 
\begin{equation} \label{Eq:L_grad}
	\nabla_{\hat{\theta}} L\left(Y,\hat{Y}\right) = \nabla_{\hat{Y}}D_{\phi}\left(Y,\hat{Y}\right) \nabla_{\hat{\theta}}\hat{Y}\left(\hat{\theta}\right)
	= \tilde{Y}^{T} \nabla^{2}\phi\left(\hat{Y}\right) \mathcal{M}\left[\Omega\right]
\end{equation}
where $\tilde{Y}\left(t\right) := \hat{Y}\left(t\right)-Y\left(t\right)$. Considering Eqs. \eqref{Eq:V_dot} and \eqref{Eq:L_grad}, the adaptation law inspired by the first-order (in time) gradient flow can be designed as
\begin{equation} \label{Eq:theta_hat_dot}
	\dot{\hat{\theta}} = - \left[\nabla^{2}\psi\left(\hat{\theta}\right)\right]^{-1} \left[\nabla_{\hat{\theta}}\left\{\dot{V}_{e} +  L\left(Y,\hat{Y}\right) \right\} \right]^{T}
	=- \left[\nabla^{2}\psi\left(\hat{\theta}\right)\right]^{-1} \left[\Phi{\mathbf{e}_{I}^{l}}^{T}PB_{e}b +  \mathcal{M}\left[\Omega\right]^{T}\nabla^{2}\phi\left(\hat{Y}\right)\tilde{Y} \right]
\end{equation}

\begin{thm}
	Consider the closed-loop system consisting of $\Sigma_{e}^{CL}$ given by Eq. \eqref{Eq:sys_e_CL} and the adaptation law given by Eq. \eqref{Eq:theta_hat_dot} where the uncertainty is modelled as in Eq. \eqref{Eq:Delta_linear}.
	Suppose that 
	\begin{enumerate}[label=\roman*)]
		\item the uncertainty basis function $\Phi\left(x\left(t\right)\right)$ is bounded
		\item $\frac{d}{dt}\mathcal{M}\left[\Omega\right]$ and $\frac{d}{dt}\nabla^{2}\phi\left(\hat{Y}\right)$ are bounded with appropriate choice of the operator $\mathcal{M}\left[\cdot\right]$ and the potential function $\phi\left(\cdot\right)$, respectively
	\end{enumerate}
	then $\mathbf{e}_{I}^{l}\left(t\right) \rightarrow 0$ as $t\rightarrow \infty$. Furthermore, provided in addition to the above conditions that
	\begin{enumerate}[label=\roman*)]
		\setcounter{enumi}{2}		
		\item the richness of the feature signal $\Phi\left(x\left(t\right)\right)$ verifies existence of $t_{ex}$ such that $\mathcal{M}\left[\Omega\right]$ attains full column rank for $\forall t \geq t_{ex}$ with appropriate choice of the regressor filter $\Sigma_{\Phi f}$ and the feature extender $\Sigma_{\Phi e}$
		\item $\nabla^{2}\phi\left(\cdot\right)>0$ with appropriate choice of the potential function $\phi\left(\cdot\right)$
	\end{enumerate}
	then $\left(\mathbf{e}_{I}^{l}\left(t\right), \tilde{\theta}\left(t\right)\right) \rightarrow \left(0,0\right)$ as $t\rightarrow \infty$.
\end{thm}
	
\begin{proof}
Along with the facts that $Q>0$ and $\phi\left(\cdot\right)$ is convex, substituting Eq. \eqref{Eq:theta_hat_dot} into Eq. \eqref{Eq:V_dot} yields
\begin{equation} \label{Eq:V_dot_adapted}
	\begin{aligned}
		\dot{V} &= -\frac{1}{2}{\mathbf{e}_{I}^{l}}^{T}Q\mathbf{e}_{I}^{l} -  \tilde{\theta}^{T} \mathcal{M}\left[\Omega\right]^{T}\nabla^{2}\phi\left(\hat{Y}\right)\tilde{Y} 
		=-\frac{1}{2}{\mathbf{e}_{I}^{l}}^{T}Q\mathbf{e}_{I}^{l} -  \tilde{\theta}^{T} \mathcal{M}\left[\Omega\right]^{T}\nabla^{2}\phi\left(\hat{Y}\right)\mathcal{M}\left[\Omega\right]\tilde{\theta}\\
		&\leq-\frac{1}{2}{\mathbf{e}_{I}^{l}}^{T}Q\mathbf{e}_{I}^{l}\leq 0
	\end{aligned}
\end{equation}
Thus, $0\leq V\left(t\right) \leq V\left(t_{0}\right)$, $\forall t\geq t_{0}$, which indicates the boundedness of $\mathbf{e}_{I}^{l}$ and $\tilde{\theta}$. Under condition i), the signals $\Phi_{f}$, $\tilde{x}$, $\Omega$, $\dot{\Omega}$, $\eta$, $\dot{\eta}$, $\dot{\mathbf{e}}_{I}^{l}=\left(A_{e}-B_{e}K_{e}\right)\mathbf{e}_{I}^{l}+B_{e}b\Phi^{T}\tilde{\theta}$ and $\dot{\hat{\theta}}$ given by Eq. \eqref{Eq:theta_hat_dot} are bounded. Also, $\ddot{V}\left(t\right)$ is bounded under condition ii). In this setting, the Barbalat's lemma ensures that $\dot{V}\left(t\right) \rightarrow 0$ as $t\rightarrow \infty$. Consequently, the squeeze theorem leads to the asymptotic convergence of tracking error, i.e., $\mathbf{e}_{I}^{l}\left(t\right) \rightarrow 0$ as $t\rightarrow \infty$.

Furthermore, conditions iii) and iv) together ensure that $\mathcal{M}\left[\Omega\right]^{T}\nabla^{2}\phi\left(\hat{Y}\right)\mathcal{M}\left[\Omega\right]>0$, $\forall t \geq t_{ex}$. As a result, $\dot{V}\left(t\right)<0$ for $\forall t \geq t_{ex}$, hence both tracking and model learning errors converge to zero, i.e., $\left(\mathbf{e}_{I}^{l}\left(t\right), \tilde{\theta}\left(t\right)\right) \rightarrow \left(0,0\right)$ as $t\rightarrow \infty$.
\end{proof}

\begin{remapp}
	The gradient flow given in Eq. \eqref{Eq:theta_hat_dot} minimises the composite loss function $L_{total}=\dot{V}_{e} +  L$. Inclusion of $\dot{V}_{e}$ as one of the loss function resembles the velocity gradient approach \cite{Andrievsky2021,Fradkov2022b}. Also, considering the connection between robustifying modification in adaptive control and regularisation in machine learning as highlighted in \cite{Gaudio2019}, one may obtain robust adaptation laws that resemble $\mu$/$e$-modification by including a suitable strictly convex regularisation term in the composite loss function.
\end{remapp}

\begin{remapp}
	The \emph{Lyapunov function} $V_{\tilde{\theta}}$ is introduced mainly for stability analysis and does not have to be measurable, whereas the \emph{regression loss function} $L$ is introduced as a performance index for optimisation and hence its gradient needs to be evaluated. The choice adopting Bregman divergence as $V_{\tilde{\theta}}$ in Eq. \eqref{Eq:V_theta} and as $L$ in Eq. \eqref{Eq:theta_hat_dot} is motivated by the recent results inspired by natural gradient descent and mirror descent such as \cite{Lee2018, Boffi2021} to allow for generalisation of the first-order adaptation laws beyond those resulting from quadratic functions.
\end{remapp}

\begin{remapp}
	The weighting/scaling factors in the potential functions $\psi\left(\cdot\right)$ and $\phi\left(\cdot\right)$ introduce the learning rates to the adaptation law. For example, defining $\psi\left(x\right)=\frac{1}{2}x^{T}\Gamma^{-1}x$ with $\Gamma>0$ gives $\left[\nabla^{2}\psi\left(\hat{\theta}\right) \right]^{-1}=\Gamma$.
\end{remapp}

\bibliographystyle{new-aiaa-nhcho}
\bibliography{SyncAC}

\begin{thebibliography}{31}
\newcommand{\enquote}[1]{``#1''}
\providecommand{\natexlab}[1]{#1}
\providecommand{\url}[1]{#1}
\providecommand{\urlprefix}{}
\expandafter\ifx\csname urlstyle\endcsname\relax
  \providecommand{\doi}[1]{\discretionary{}{}{}https://doi.org/#1}\else
  \providecommand{\doi}[1]{\discretionary{}{}{}\urlstyle{rm}\href{https://doi.org/#1}{doi:#1}}\fi

\bibitem[{Lavretsky(2009)}]{Lavretsky2009}
Lavretsky, E., \enquote{Combined/Composite Model Reference Adaptive Control,} \emph{IEEE Transactions on Automatic Control}, Vol.~54, No.~11, 2009, pp. 2692--2697.
\newblock \doi{10.1109/TAC.2009.2031580}.

\bibitem[{Hovakimyan and Cao(2010)}]{Hovakimyan2010}
Hovakimyan, N., and Cao, C., \emph{$\mathcal{L}_{1}$ Adaptive Control Theory: Guaranteed Robustness with Fast Adaptation}, Advances in Design and Control, Society for Industrial and Applied Mathematics, Philadelphia, PA, USA, 2010.
\newblock \doi{10.1137/1.9780898719376}.

\bibitem[{Ortega et~al.(2020)Ortega, Nikiforov, and Gerasimov}]{Ortega2020}
Ortega, R., Nikiforov, V., and Gerasimov, D., \enquote{On Modified Parameter Estimators for Identification and Adaptive Control. A Unified Framework and Some New Schemes,} \emph{Annual Reviews in Control}, Vol.~50, 2020, pp. 278--293.
\newblock \doi{10.1016/j.arcontrol.2020.06.002}.

\bibitem[{Lavretsky and Wise(2013)}]{Lavretsky2013}
Lavretsky, E., and Wise, K.~A., \emph{Robust and Adaptive Control: With Aerospace Applications}, Springer-Verlag, London, UK, 2013, Chap.~13.
\newblock \doi{10.1007/978-1-4471-4396-3}.

\bibitem[{Gibson et~al.(2013)Gibson, Annaswamy, and Lavretsky}]{Gibson2013}
Gibson, T.~E., Annaswamy, A.~M., and Lavretsky, E., \enquote{On Adaptive Control with Closed-Loop Reference Models: Transients, Oscillations, and Peaking,} \emph{IEEE Access}, Vol.~1, 2013, pp. 703--717.
\newblock \doi{10.1109/ACCESS.2013.2284005}.

\bibitem[{Gibson(2014)}]{Gibson2014}
Gibson, T.~E., \enquote{Closed-Loop Reference Model Adaptive Control: With Application to Very Flexible Aircraft,} Ph.D. thesis, Department of Mechanical Engineering, Massachusetts Institute of Technology, Cambridge, MA, USA, 2014.
\newblock \urlprefix\url{http://hdl.handle.net/1721.1/87974}.

\bibitem[{Gibson et~al.(2015)Gibson, Qu, Annaswamy, and Lavretsky}]{Gibson2015}
Gibson, T.~E., Qu, Z., Annaswamy, A.~M., and Lavretsky, E., \enquote{Adaptive Output Feedback Based on Closed-Loop Reference Models,} \emph{IEEE Transactions on Automatic Control}, Vol.~60, No.~10, 2015, pp. 2728--2733.
\newblock \doi{10.1109/TAC.2015.2405295}.

\bibitem[{Qu et~al.(2020)Qu, Thomsen, and Annaswamy}]{Qu2020}
Qu, Z., Thomsen, B.~T., and Annaswamy, A.~M., \enquote{Adaptive Control for a Class of Multi-Input Multi-Output Plants with Arbitrary Relative Degree,} \emph{IEEE Transactions on Automatic Control}, Vol.~65, No.~7, 2020, pp. 3023--3038.
\newblock \doi{10.1109/TAC.2019.2941420}.

\bibitem[{Wiese et~al.(2015)Wiese, Annaswamy, Muse, Bolender, and Lavretsky}]{Wiese2015}
Wiese, D.~P., Annaswamy, A.~M., Muse, J.~A., Bolender, M.~A., and Lavretsky, E., \enquote{Adaptive Output Feedback Based on Closed-Loop Reference Models for Hypersonic Vehicles,} \emph{Journal of Guidance, Control, and Dynamics}, Vol.~38, No.~12, 2015, pp. 2429--2440.
\newblock \doi{10.2514/1.G001098}.

\bibitem[{Qu and Annaswamy(2016)}]{Qu2016}
Qu, Z., and Annaswamy, A.~M., \enquote{Adaptive Output-Feedback Control with Closed-Loop Reference Models for Very Flexible Aircraft,} \emph{Journal of Guidance, Control, and Dynamics}, Vol.~39, No.~4, 2016, pp. 873--888.
\newblock \doi{10.2514/1.G001282}.

\bibitem[{Yuksek and Inalhan(2021)}]{Yuksek2021}
Yuksek, B., and Inalhan, G., \enquote{Reinforcement Learning Based Closed-Loop Reference Model Adaptive Flight Control System Design,} \emph{International Journal of Adaptive Control and Signal Processing}, Vol.~35, No.~3, 2021, pp. 420--440.
\newblock \doi{10.1002/acs.3181}.

\bibitem[{Ren and Beard(2005)}]{Ren2005a}
Ren, W., and Beard, R.~W., \enquote{Consensus Seeking in Multiagent Systems Under Dynamically Changing Interaction Topologies,} \emph{IEEE Transactions on Automatic Control}, Vol.~50, No.~5, 2005, pp. 655--661.
\newblock \doi{10.1109/TAC.2005.846556}.

\bibitem[{Ren et~al.(2005)Ren, Beard, and Atkins}]{Ren2005b}
Ren, W., Beard, R.~W., and Atkins, E.~M., \enquote{A Survey of Consensus Problems in Multi-Agent Coordination,} \emph{American Control Conference}, Portland, OR, USA, 2005.
\newblock \doi{10.1109/ACC.2005.1470239}.

\bibitem[{Ren et~al.(2007)Ren, Beard, and Atkins}]{Ren2007}
Ren, W., Beard, R.~W., and Atkins, E.~M., \enquote{Information Consensus in Multivehicle Cooperative Control,} \emph{IEEE Control Systems Magazine}, Vol.~27, No.~2, 2007, pp. 71--82.
\newblock \doi{10.1109/MCS.2007.338264}.

\bibitem[{Oh et~al.(2015)Oh, Park, and Ahn}]{Oh2015}
Oh, K.-K., Park, M.-C., and Ahn, H.-S., \enquote{A Survey of Multi-Agent Formation Control,} \emph{Automatica}, Vol.~53, 2015, pp. 424--440.
\newblock \doi{10.1016/j.automatica.2014.10.022}.

\bibitem[{Kim et~al.(2016)Kim, Yang, Shim, Kim, and Seo}]{Kim2016}
Kim, J., Yang, J., Shim, H., Kim, J.-S., and Seo, J.~H., \enquote{Robustness of Synchronization of Heterogeneous Agents by Strong Coupling and a Large Number of Agents,} \emph{IEEE Transactions on Automatic Control}, Vol.~61, No.~10, 2016, pp. 3096--3102.
\newblock \doi{10.1109/TAC.2015.2498138}.

\bibitem[{Lee and Shim(2020)}]{Lee2020}
Lee, J.~G., and Shim, H., \enquote{A Tool for Analysis and Synthesis of Heterogeneous Multi-Agent Systems under Rank-Deficient Coupling,} \emph{Automatica}, Vol. 117, No. 108952, 2020.
\newblock \doi{10.1016/j.automatica.2020.108952}.

\bibitem[{Lee et~al.(2022)Lee, Trenn, and Shim}]{Lee2022b}
Lee, J.~G., Trenn, S., and Shim, H., \enquote{Synchronization with Prescribed Transient Behavior: Heterogeneous Multi-agent Systems under Funnel Coupling,} \emph{Automatica}, Vol. 141, No. 110276, 2022.
\newblock \doi{10.1016/j.automatica.2022.110276}.

\bibitem[{Lee and Shim(2022{\natexlab{a}})}]{Lee2022a}
Lee, J.~G., and Shim, H., \emph{Design of Heterogeneous Multi-agent System for Distributed Computation}, Springer, Cham, Switzerland, 2022{\natexlab{a}}, Chap.~4, pp. 83--108.
\newblock \doi{10.1007/978-3-030-74628-5}.

\bibitem[{Lee and Shim(2022{\natexlab{b}})}]{Lee2022c}
Lee, S., and Shim, H., \enquote{Blended Dynamics Approach to Distributed Optimization: Sum Convexity and Convergence Rate,} \emph{Automatica}, Vol. 141, No. 110290, 2022{\natexlab{b}}.
\newblock \doi{10.1016/j.automatica.2022.110290}.

\bibitem[{Cho(2023)}]{Cho2023}
Cho, N., \enquote{SyncAC: Julia-Based Simulation Code for Synchronisation-Oriented Adaptive Control,} , March 2023.
\newblock \urlprefix\url{https://github.com/nhcho91/SyncAC}.

\bibitem[{Gahlawat et~al.(2020)Gahlawat, Zhao, Patterson, Hovakimyan, and Theodorou}]{Gahlawat2020}
Gahlawat, A., Zhao, P., Patterson, A., Hovakimyan, N., and Theodorou, E., \enquote{$\mathcal{L}_{1}$-$\mathcal{GP}$: $\mathcal{L}_{1}$ Adaptive Control with Bayesian Learning,} \emph{2nd Conference on Learning for Dynamics and Control}, Virtual (Berkeley, CA, USA), 2020, pp. 826--837.
\newblock \urlprefix\url{https://proceedings.mlr.press/v120/gahlawat20a.html}.

\bibitem[{Chowdhary et~al.(2015)Chowdhary, Kingravi, How, and Vela}]{Chowdhary2015}
Chowdhary, G., Kingravi, H.~A., How, J.~P., and Vela, P.~A., \enquote{Bayesian Nonparametric Adaptive Control Using Gaussian Processes,} \emph{IEEE Transactions on Neural Networks and Learning Systems}, Vol.~26, No.~3, 2015, pp. 537--550.
\newblock \doi{10.1109/TNNLS.2014.2319052}.

\bibitem[{Joshi and Chowdhary(2019)}]{Joshi2019}
Joshi, G., and Chowdhary, G., \enquote{Deep Model Reference Adaptive Control,} \emph{58th IEEE Conference on Decision and Control}, Nice, France, 2019.
\newblock \doi{10.1109/CDC40024.2019.9029173}.

\bibitem[{Joshi et~al.(2020)Joshi, Virdi, and Chowdhary}]{Joshi2020}
Joshi, G., Virdi, J., and Chowdhary, G., \enquote{Asynchronous Deep Model Reference Adaptive Control,} \emph{4th Conference on Robot Learning}, Cambridge, MA, USA, 2020, pp. 984--1000.
\newblock \urlprefix\url{https://proceedings.mlr.press/v155/joshi21a.html}.

\bibitem[{Boffi et~al.(2022)Boffi, Tu, and Slotine}]{Boffi2022}
Boffi, N.~M., Tu, S., and Slotine, J.-J.~E., \enquote{Nonparametric Adaptive Control and Prediction: Theory and Randomized Algorithms,} \emph{Journal of Machine Learning Research}, Vol.~23, No. 281, 2022, pp. 1--46.
\newblock \urlprefix\url{http://jmlr.org/papers/v23/22-0022.html}.

\bibitem[{Andrievsky and Fradkov(2021)}]{Andrievsky2021}
Andrievsky, B.~R., and Fradkov, A.~L., \enquote{Speed Gradient Method and Its Applications,} \emph{Automation and Remote Control}, Vol.~82, No.~9, 2021, pp. 1463--1518.
\newblock \doi{10.1134/S0005117921090010}.

\bibitem[{Fradkov and Andrievsky(2022)}]{Fradkov2022b}
Fradkov, A.~L., and Andrievsky, B.~R., \enquote{Speed-Gradient Method in Mechanical Engineering,} \emph{Mechanics and Control of Solids and Structures}, Springer, Cham, Switzerland, 2022, pp. 171--194.
\newblock \doi{10.1007/978-3-030-93076-9\_9}.

\bibitem[{Gaudio et~al.(2019)Gaudio, Gibson, Annaswamy, Bolender, and Lavretsky}]{Gaudio2019}
Gaudio, J.~E., Gibson, T.~E., Annaswamy, A.~M., Bolender, M.~A., and Lavretsky, E., \enquote{Connections Between Adaptive Control and Optimization in Machine Learning,} \emph{IEEE Conference on Decision and Control}, Nice, France, 2019, pp. 4563--4568.
\newblock \doi{10.1109/CDC40024.2019.9029197}.

\bibitem[{Lee et~al.(2018)Lee, Kwon, and Park}]{Lee2018}
Lee, T., Kwon, J., and Park, F.~C., \enquote{A Natural Adaptive Control Law for Robot Manipulators,} \emph{IEEE/RSJ International Conference on Intelligent Robots and Systems}, Madrid, Spain, 2018.
\newblock \doi{10.1109/IROS.2018.8593727}.

\bibitem[{Boffi and Slotine(2021)}]{Boffi2021}
Boffi, N.~M., and Slotine, J.-J.~E., \enquote{Implicit Regularization and Momentum Algorithms in Nonlinearly Parameterized Adaptive Control and Prediction,} \emph{Neural Computation}, Vol.~33, No.~3, 2021, pp. 590--673.
\newblock \doi{10.1162/neco\_a\_01360}.

\end{thebibliography}
\end{document}